\documentclass[11pt,reqno]{article}
\usepackage{latexsym, amsmath, amssymb, bm, a4, epsfig,cite}
\usepackage{amsthm}
\usepackage{graphicx}
\usepackage{caption}
\usepackage{subcaption}
\usepackage{booktabs}

\newtheorem{thm}{Theorem}[section]
\newtheorem{lem}[thm]{Lemma}

\newtheorem{prop}[thm]{Proposition}

\usepackage[usenames]{color}
\usepackage[colorlinks=true,urlcolor=blue,linkcolor=blue,citecolor=red]{hyperref}
\usepackage[left=1in,right=1in,top=1in,bottom=1in]{geometry}

\newcommand{\bA}{\mathbf{A}}
\newcommand{\bB}{\mathbf{B}}
\newcommand{\bE}{\mathbf{E}}
\newcommand{\bF}{\mathbf{F}}
\newcommand{\bG}{\mathbf{G}}
\newcommand{\bH}{\mathbf{H}}
\newcommand{\bI}{\mathbf{I}}

\newcommand{\bPhi}{\mathbf{\Phi}}
\newcommand{\bPsi}{\mathbf{\Psi}}

\newcommand{\tbE}{\tilde{\mathbf{E}}}

\newcommand*\diff{\mathop{}\!\mathrm{d}}

\renewcommand{\div}{{\rm div} \,}
\newcommand{\curl}{{\rm curl} \,}

\newcommand{\beq}{\begin{equation}}
\newcommand{\eeq}{\end{equation}}

\title{Electromagnetic interior transmission eigenvalue problem for inhomogeneous media containing obstacles and its applications to near cloaking}

\author{Jingzhi Li\thanks{Faculty of Science, South University of Science and Technology of China, 518055, Shenzhen, P. R. China (lijz@sustc.edu.cn, lixf@sustc.edu.cn).} \and Xiaofei Li\footnotemark[2]  \and Hongyu Liu\thanks{\footnotesize Department of Mathematics, Hong Kong Baptist University, Kowloon Tong, Hong
Kong SAR (hongyu.liuip@gmail.com, jadelightking@qq.com).}\and Yuliang Wang\footnotemark[3] }

\begin{document}

\maketitle

\begin{abstract}
This paper is concerned with the invisibility cloaking in electromagnetic wave scattering from a new perspective. We are especially interested in achieving the invisibility cloaking by completely regular and isotropic mediums. Our study is based on an interior transmission eigenvalue problem. We propose a cloaking scheme that takes a three-layer structure including a cloaked region, a lossy layer and a cloaking shell. The target medium in the cloaked region can be arbitrary but regular, whereas the mediums in the lossy layer and the cloaking shell are both regular and isotropic. We establish that there exists an infinite set of incident waves such that the cloaking device is nearly-invisible under the corresponding wave interrogation. The set of waves is generated from the Maxwell-Herglotz approximation of the associated interior transmission eigenfunctions. We provide the mathematical design of the cloaking device and sharply quantify the cloaking performance. 
\end{abstract}



\section{Introduction}
Invisibility cloaking has received significant attentions in recent years in the scientific community due to its practical importance; see \cite{AKLL1,AKLL2,AKLL3,BL,BLZ,BPS,GKLU,GKLU2,GLU,KSVW,JL,LWZ} and the references therein for the relevant mathematical literature. The crucial idea is to coat a target object with a layer of artificially engineered material with desired properties so that the electromagnetic waves pass through the device without creating any shadow at the other end; namely, invisibility cloaking is achieved. Invisibility cloaking could find striking applications in many areas of science and technology such as radar and sonar, medical imaging, earthquake science and, energy science and engineering, to name just a few. 

Generally speaking, a region of space is said to be cloaked if its contents, together with the cloak, are invisible to a particular class of wave measurements. In the literature, most of the existing works are concerned with the design of certain artificial mechanisms of controlling wave propagation so that invisibility is achieved independent of the source of the detecting waves; that is, for any generic wave fields that one uses to impinge on the cloaking device, there will be invisibility effect produced. In this paper, we shall develop a novel cloaking scheme where the invisibility is only achieved with respect to detecting waves from a particular set. In doing so, one can achieve the invisibility cloaking by completely regular and isotropic mediums. Next, we first present the mathematical setup and then discuss the main results of the current study. 

Consider the time-harmonic electromagnetic (EM) wave scattering in a homogeneous space with the presence of an inhomogeneous scatterer. Let us first characterize the optical properties of an EM medium with the electric permittivity $\epsilon$, magnetic permeability $\mu$, and electric conductivity $\sigma$. We recall that $\mathbb{M}^{3\times 3}_{sym}$ is the space of real-valued symmetric matrices and that, for any Lipschitz domain $\Omega\subset \mathbb{R}^3$, we say that $\gamma$ is a tensor in $\Omega$ satisfying the uniform ellipticity condition if $\gamma\in L^{\infty}(\Omega;\mathbb{M}^{3\times 3}_{sym})$ and there exists $0<c_0<1$ such that
$$c_0|\xi|^2\leq \gamma(x)\xi\cdot\xi\leq c_0^{-1}|\xi|^2~\mbox{for}~ {\it a.e.}~x\in \Omega~\mbox{and every}~\xi\in\mathbb{R}^3.$$
$c_0$ shall be referred to as the ellipticity constant of the tensor $\gamma$. It is assumed that both $\epsilon(x)$ and $\mu(x)$, $x\in \mathbb{R}^3$, belong to $L^{\infty}(\Omega;\mathbb{M}^{3\times 3}_{sym})$, and are uniform elliptic with constant $c_0\in \mathbb{R}_+$; whereas it is also assumed that $\sigma\in L^{\infty}(\Omega;\mathbb{M}^{3\times 3}_{sym})$ satisfies
$$0\leq \sigma(x)\xi\cdot\xi\leq \lambda_0|\xi|^2~\mbox{for}~ {\it a.e.}~x\in \mathbb{R}^3~\mbox{and every}~\xi\in\mathbb{R}^3,$$
where $\lambda_0\in\mathbb{R}^+$. Denote $(\Omega;\epsilon,\mu,\sigma)$ the medium $\Omega$ associated with $\epsilon,\mu,\sigma$, and it is said to be regular if the material parameters fulfill the conditions described above. Moreover, $\gamma(x)$, $x\in \Omega$, is said to be isotropic if there exists $\alpha(x)\in L^{\infty}(\Omega;\mathbb{R})$ such that $\gamma(x)=\alpha(x)\cdot \bI_{3\times 3}$, where $\bI_{3\times 3}$ signifies the $3\times 3$ identity matrix. Suppose $(\Omega;\epsilon,\mu,\sigma)$ is located in an isotropic and homogeneous background/matrix medium whose material parameters are given by $$\epsilon(x)=\mathbf{I}_{3\times 3},~\mu(x)=\mathbf{I}_{3\times 3},~\sigma(x)=0,~\mbox{for}~x\in\mathbb{R}^3\backslash{\bar{\Omega}}.$$

Let $\omega\in\mathbb{R}^+$ denote an EM wavenumber, corresponding to a certain EM spectrum. Consider the EM radiation in this frequency regime in the space
$$(\mathbb{R}^3;\epsilon,\mu,\sigma)=(\Omega;\epsilon,\mu,\sigma) \wedge (\mathbb{R}^3\backslash\bar{\Omega}; \mathbf{I}_{3 \times 3}, \mathbf{I}_{3 \times 3}, 0).$$
Let $(\bE^i,\bH^i)$ be a pair of entire electric and magnetic fields, modeling the illumination source. They verify the time-harmonic Maxwell equations,
\beq
\curl \bE^i-i\omega \bH^i=0,~~\curl \bH^i+i\omega \bE^i=0~~\mbox{in}~\mathbb{R}^3.
\label{Maxwell}
\eeq
The presence of the inhomogeneous scatterer $(\Omega;\epsilon,\mu,\sigma)$ interrupts the propagation of the EM waves $\bE^i$ and $\bH^i$, leading to the so-called wave scattering. We let $\bE^s$ and $\bH^s$ denote, respectively, the scattered electric and magnetic fields. Define
\begin{equation*}
\bE:=\bE^i+\bE^s,~~\bH:=\bH^i+\bH^s,
\end{equation*}
to be the total electric and magnetic fields, respectively. Then the EM scattering is governed by the following Maxwell system
\beq
\begin{cases}
\curl\bE(x)-i\omega \mu(x)\bH(x)=0,~&x\in\mathbb{R}^3,\\
\curl\bH(x)+i\omega \epsilon(x)\bE(x)=\sigma(x)\bE(x),~&x\in\mathbb{R}^3,\\
\lim_{|x|\rightarrow +\infty}\big(\mu^{1/2}(x)\bH^s(x)\times x-|x|\epsilon ^{1/2}(x)\bE^s(x)\big)=0.
\end{cases}
\label{max}
\eeq
The last limit in (\ref{max}) is known as the Silver-M$\ddot{\mbox{u}}$ller radiation condition. The Maxwell system (\ref{max}) is well-posed and there exists a unique pair of solutions $(\bE,\bH)\in H(\mbox{curl},\mathbb{R}^3)$. Here and also in what follows, for any open set $\Omega\subset\mathbb{R}^3$, we make use of the following Sobolev spaces:
\begin{equation}
\begin{split}
H(\mbox{curl},\Omega)&:=\left \{u\in L^2(\Omega)^3|\curl u\in L^2(\Omega)^3\right \},\\
H_0(\mbox{curl},\Omega)&:=\left \{u\in H(\mbox{curl},\Omega): \nu\times u=0,~\nu\times\curl u=0~\mbox{on}~\partial\Omega\right \},\\
H^2(\mbox{curl},\Omega)&:=\left \{u\in H(\mbox{curl},\Omega):\curl u\in H(\mbox{curl},\Omega)\right \},\\
H_0^2(\mbox{curl},\Omega)&:=\left \{u\in H_0(\mbox{curl},\Omega):\curl u\in H(\mbox{curl},\Omega)\right \},
\end{split}
\label{Hcurl}
\end{equation}
endowed with the scalar product 
\begin{align*}
(u,v)_{H(\mbox{curl},\Omega)} &= (u,v)_{L^2(\Omega)}+(\curl u,\curl v)_{L^2(\Omega)}, \\ 
(u,v)_{H^2(\mbox{curl},\Omega)} &= (u,v)_{H(\mbox{curl},\Omega)}+(\curl u,\curl v)_{H(\mbox{curl},\Omega)},
\end{align*}
and the corresponding norms $\|\cdot\|_{{H(\mbox{curl},\Omega)}}$ and $\|\cdot\|_{{H^2(\mbox{curl},\Omega)}}$. 
Moreover, define
$$TH^{-1/2}_{{\rm Div}}(\partial\Omega):=\big\{ U\in TH^{-1/2}(\partial\Omega):{\rm Div}\, U\in H^{-1/2}(\partial\Omega)\big\},$$
where Div is the surface divergence operator on $\partial\Omega$, $TH^{s}(\partial\Omega)$ is the subspace of all those $V\in (H^s(\partial\Omega))^3$ which are orthogonal to $\nu$, and $H^s(\cdot)$ is the usual $L^2$-based Sobolev space of order $s\in \mathbb{R}$. 

For the solutions to (\ref{max}), we have that as $|x|\rightarrow +\infty$\cite{CK,JCN}:
$$\bE^s(x)=\frac{e^{i\omega|x|}}{|x|}\bE_{\infty}(\hat{x})+O\left(\frac{1}{|x|^2}\right),$$
$$\bH^s(x)=\frac{e^{i\omega|x|}}{|x|}\bH_{\infty}(\hat{x})+O\left(\frac{1}{|x|^2}\right),$$
where $\hat{x}:=x/|x|\in\mathbb{S}^2$, $x\in\mathbb{R}^3\backslash\{0\}$. $\bE_{\infty}$ and $\bH_{\infty}$ are, respectively, referred to as the electric and magnetic far-field patterns, and they satisfy
$$\bH_{\infty}(\hat{x})=\hat{x}\times \bE_{\infty}(\hat{x})~~\mbox{and}~~\hat{x}\cdot \bE_{\infty}(\hat{x})=\hat{x}\cdot \bH_{\infty}(\hat{x})=0,~~ \forall\,\hat{x}\in\mathbb{S}^2.$$
The medium $(\Omega;\epsilon,\mu,\sigma)$ is said to be invisible under the electromagnetic wave interrogation by $(\bE^i,\bH^i)$ if
\beq
\bE_{\infty}(\hat{x},(\bE^i,\bH^i),(\Omega;\epsilon,\mu,\sigma))\equiv 0~\mbox{and}~\bH_{\infty}(\hat{x},(\bE^i,\bH^i),(\Omega;\epsilon,\mu,\sigma))\equiv 0,~~\forall\, \hat{x}\in \mathbb{S}^2.
\label{invi}
\eeq

In the current work, we shall consider the cloaking technique in achieving the invisibility. Let $D\Subset \Omega$ be a bounded Lipschitz domain. Consider a cloaking device of the following form
\beq
(\mathbb{R}^3;\epsilon,\mu,\sigma)=(D;\epsilon_c,\mu_c,\sigma_c)\wedge(\Omega\backslash\bar{D};\epsilon_m,\mu_m,\sigma_m) \wedge (\mathbb{R}^3 \backslash \bar{\Omega}; \mathbf{I}_{3 \times 3}, \mathbf{I}_{3 \times 3}, 0),
\label{device}
\eeq
where $(D;\epsilon_c,\mu_c,\sigma_c)$ denotes the target object being cloaked, and $(\Omega\backslash\bar{D};\epsilon_m,\mu_m,\sigma_m)$ denotes the cloaking shell medium (see Figure \ref{fig:two_layer}).
\begin{figure}
  \centering
  \includegraphics[width=0.3\textwidth]{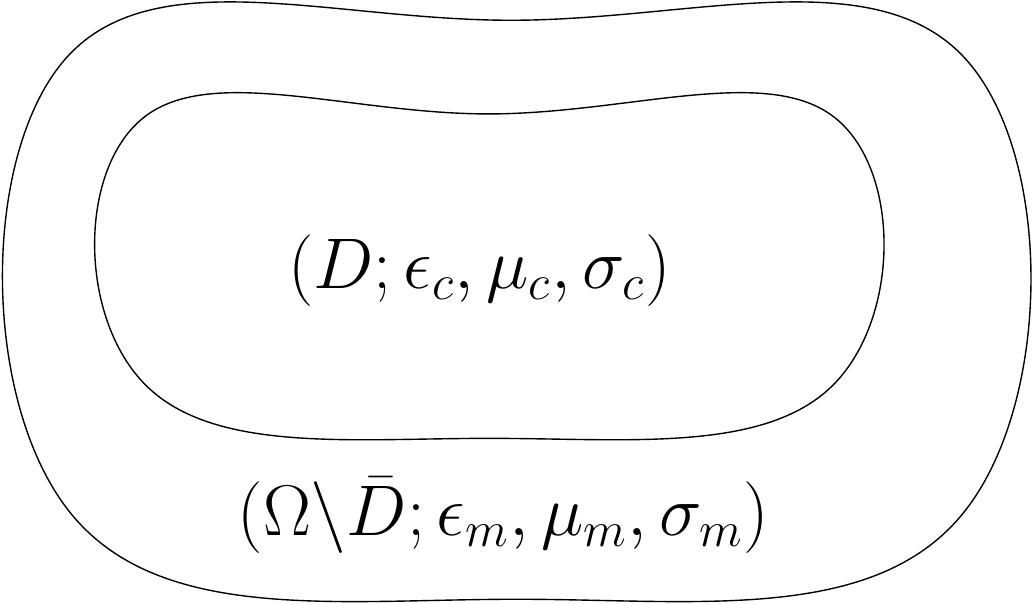}
  \caption{A two-layer cloaking device.}
  \label{fig:two_layer}
\end{figure}

For a practical cloaking device of the form (\ref{device}), there are several crucial ingredients that one should incorporate into the design:
\begin{itemize}
\item  The target object $(D;\epsilon_c,\mu_c,\sigma_c)$ can be allowed to be arbitrary (but regular). That is, the cloaking device should not be object-dependent. In what follows, this issue shall be referred to as the target independence for a cloaking device.
\item The cloaking medium $(\Omega\backslash\bar{D};\epsilon_m,\mu_m,\sigma_m)$ should be feasible for construction and fabrication. Indeed, it would be the most practically feasible if $(\Omega\backslash\bar{D};\epsilon_m,\mu_m,\sigma_m)$ is uniformly elliptic with fixed constants and isotropic as well. In what follows, this issue shall be referred to as the practical feasibility for a cloaking device.
\item  For an ideal cloaking device, one can expect the invisibility performance (\ref{invi}). However, in practice, especially in order to fulfill the above two requirements, one can relax the ideal cloaking requirement (\ref{invi}) to be
$$|\bE_{\infty}(\hat{x},(\bE^i,\bH^i),(\Omega;\epsilon,\mu,\sigma))|\ll 1~~\mbox{and}~~|\bH_{\infty}(\hat{x},(\bE^i,\bH^i),(\Omega;\epsilon,\mu,\sigma))|\ll 1,$$
$\forall\, \hat{x}\in \mathbb{S}^2$ and $\forall\, (\bE^i,\bH^i)\in \mathcal{H}$, where $\mathcal{H}$ is a set of incident waves consisting of entire solutions to the Maxwell equation (\ref{Maxwell}). That is, near-invisibility can be achieved for scattering measurements made with interrogating waves from the set $\mathcal{H}$. In what follows, this issue shall be referred to as the relaxation and approximation for a cloaking device.
\end{itemize}

In this paper, we shall develop a cloaking scheme that addresses all of the issues discussed above. Our study connects to a so-called interior transmission eigenvalue problem associated with the Maxwell system as follows, 
\begin{equation}
\begin{cases}
\curl \bE_m(x)-i\omega \mu_m\bH_m(x)=0 \quad &~\mbox{in}~ \Omega\backslash \bar{D},\\
\curl \bH_m(x)+i\omega \epsilon_m\bE_m(x)=0 \quad &~\mbox{in}~ \Omega\backslash \bar{D},\\
\curl \bE_0(x)-i\omega\bH_0(x)=0 \quad &~\mbox{in} ~\Omega,\\
\curl \bH_0(x)+i\omega\bE_0(x)=0 \quad &~\mbox{in} ~\Omega,\\
\nu\times \bE_m =0 \quad &~\mbox{on}~\partial D,\\
\nu\times \bE_m=\nu\times \bE_0,~~\nu\times\bH_m=\nu\times\bH_0 \quad &~\mbox{on} ~\partial\Omega.
\end{cases}
\label{u0}
\end{equation} 
Concerning the problem \eqref{u0}, some remarks are in order. The interior transmission eigenvalue problem is associated with an isotropic EM medium $(\Omega\backslash\bar{D};\epsilon_m,\mu_m,0)$ that contains a PEC (perfectly electric conducting) obstacle $D$. A similar interior transmission eigenvalue problem arising from the acoustic scattering governed by the Helmholtz system associated with an isotropic acoustic medium containing an impenetrable obstacle was considered in \cite{CCH}. The corresponding result was applied to the invisibility cloaking study for acoustic waves in \cite{JL}. In the current article, we shall extend those studies to the much more technical and complicated Maxwell system governing the electromagnetic scattering. We shall first prove the discreteness and existence of the interior transmission eigenvalues of the system (\ref{u0}). To our best knowledge, those results are new to literature on the study of interior transmission eigenvalue problems. Then we shall apply the obtained results to the invisibility cloaking study. The first result we can show concerning the invisibility cloaking is as follows. 
\begin{prop}
Consider the EM configuration $(\mathbb{R}^3;\epsilon,\mu,\sigma)$ in (\ref{device}) with $\sigma_m\equiv 0$ and $(D; \epsilon_c,\mu_c,\sigma_c)$ replaced to be a PEC obstacle. Let $\omega\in\mathbb{R}^+$ be an interior transmission eigenvalue associated with $(\Omega\backslash\bar{D};\epsilon_m,\mu_m,0)$, and $(\bE_m,\bH_m)$, $(\bE_0,\bH_0)$ be a corresponding pairs of eigenfunctions of (\ref{u0}). For any sufficiently small $\varepsilon\in\mathbb{R}^+$, by the denseness property of Maxwell-Herglotz functions (cf. \eqref{Herg}), there exists $(\bE^g_{\omega},\bH^g_{\omega})$ such that
\beq
\|\bE^g_{\omega}-\bE_0\|_{H({\rm curl}, \Omega)}<\varepsilon,~~\|\bH^g_{\omega}-\bH_0\|_{H({\rm curl}, \Omega)}<\varepsilon.
\label{EH}
\eeq
Consider the scattering problem (\ref{max}) by taking the incident electric and magnetic wave fields
$$\bE^i=\bE^g_{\omega},~~\bH^i=\bH^g_{\omega},$$
then there hold
\beq
|\bE_{\infty}(\hat{x},(\bE^g_{\omega},\bH^g_{\omega}),(\Omega\backslash\bar{D};\epsilon_m,\mu_m,0))|\leq C\varepsilon,~~\forall\, \hat{x}\in \mathbb{S}^{2},
\label{nearE}
\eeq
and
\beq
 |\bH_{\infty}(\hat{x},(\bE^g_{\omega},\bH^g_{\omega}),(\Omega\backslash\bar{D};\epsilon_m,\mu_m,0))|\leq C\varepsilon,~~\forall\, \hat{x}\in \mathbb{S}^{2},
\label{nearH}
\end{equation}
where $C$ is a positive constant independent of $\varepsilon$ and $(\bE^g_{\omega},\bH^g_{\omega})$.
\label{nearcloaking}
\end{prop}
Proposition~\ref{nearcloaking} is kind of a folk-telling result in the literature on the study of interior transmission eigenvalues, and it shall be needed in our cloaking study. 

Motivated by the two-layer cloaking device we then consider a three-layer cloaking device. Let $\Sigma\Subset D$ be bounded Lipschitz domain such that $D\backslash\bar{\Sigma}$ is connected and consider an EM medium configuration as follows (see Figure \ref{fig:three_layer}),
\beq
(\mathbb{R}^3;\epsilon,\mu,\sigma) = (\Sigma;\epsilon_a,\mu_a,\sigma_a)\wedge(D\backslash\bar{\Sigma};\epsilon_l,\mu_l,\sigma_l)\wedge(\Omega\backslash\bar{D};\epsilon_m,\mu_m,0)\wedge (\mathbb{R}^3\backslash\bar{\Omega};\bI_{3\times 3},\bI_{3\times 3},0),
\label{medium0}
\eeq
where $(\Omega\backslash\bar{D};\epsilon_m,\mu_m,0)$ is isotropic and the lossy layer $(D\backslash\bar{\Sigma};\epsilon_l,\mu_l,\sigma_l)$ is chosen to be $$\epsilon_l=\alpha_1\tau^{-1}\cdot\bI_{3\times 3},~~ \mu_l=\alpha_2\tau\cdot\bI_{3\times 3},~~ \sigma_l=\alpha_3\tau^{-1}\cdot\bI_{3\times 3},$$
where $\tau\in\mathbb{R}^+$ is an asymptotically small parameter, and $\alpha_1,\alpha_2,\alpha_3$ are constants in $\mathbb{R}^+$. The target medium $(\Sigma;\epsilon_a,\mu_a,\sigma_a)$ in the cloaked region $\Sigma$ can be arbitrary but regular. 
\begin{figure}
  \centering
  \includegraphics[width=0.35\textwidth]{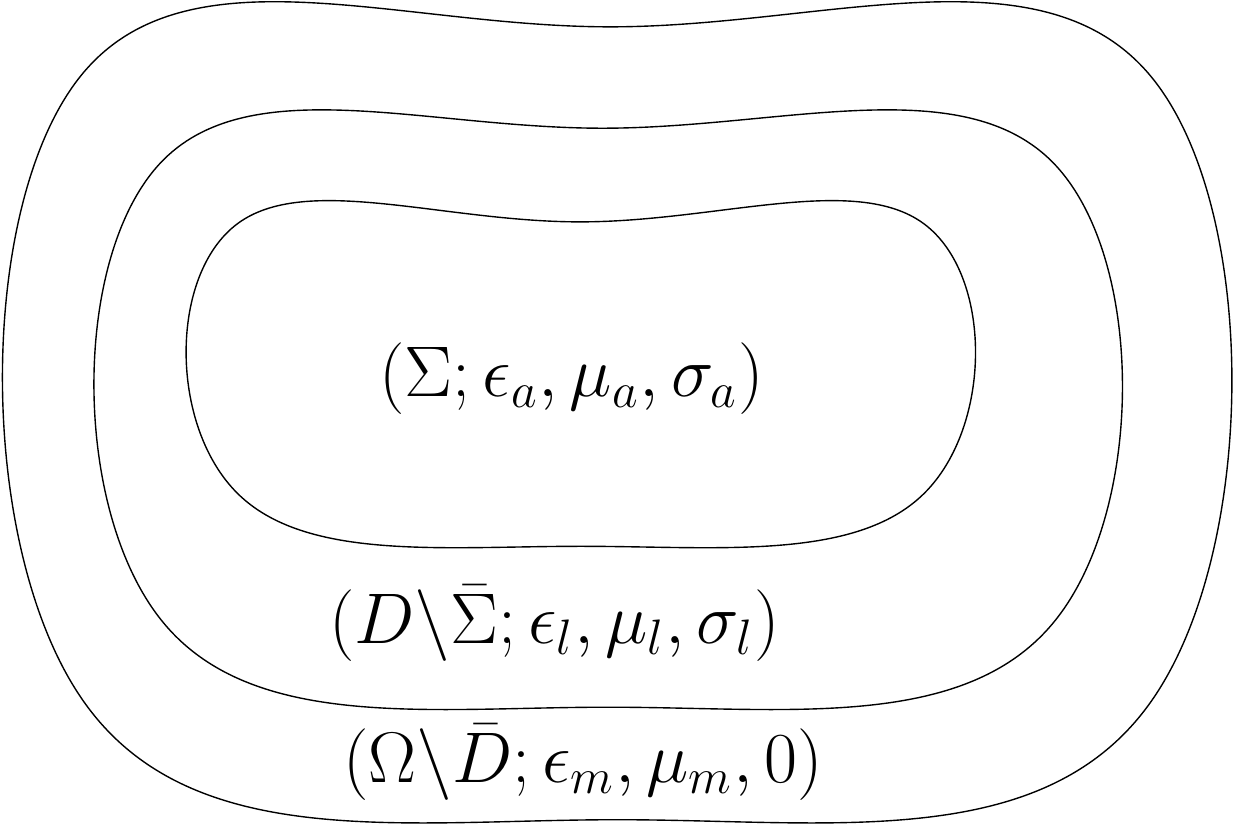}
  \caption{A three-layer cloaking device.}
  \label{fig:three_layer}
\end{figure}

For such a cloaking construction, our main theorem is as follows:
\begin{thm} \label{thm:1}
Let $(\mathbb{R}^3;\epsilon,\mu,\sigma)$ be described in (\ref{medium0}). Let $\omega$ and $(\bE^g_{\omega},\bH^g_{\omega})$ be the same as those in Proposition \ref{nearcloaking}. Consider the scattering system (\ref{max}) corresponding to $(\mathbb{R}^3;\epsilon,\mu,\sigma)$ with $\bE^i=\bE^g_{\omega},~ \bH^i=\bH^g_{\omega}$. Then $\forall\,\hat{x}\in \mathbb{S}^{2}$ we have
\begin{multline} \label{Einvi}
|\bE_{\infty}(\hat{x},(\bE^g_{\omega},\bH^g_{\omega}),(D\backslash \bar{\Sigma};\epsilon_l,\mu_l,\sigma_l) \wedge (\Omega\backslash\bar{D};\epsilon_m,\mu_m,0))| \\ \leq C\Big(\varepsilon+\tau^{1/2}\|\bE^g_{\omega}\|_{H({\rm curl},\Omega)}+\tau^{1/2}\|\bH^g_{\omega}\|_{H({\rm curl},\Omega)}\Big),
\end{multline}
\begin{multline} \label{Hinvi}
|\bH_{\infty}(\hat{x},(\bE^g_{\omega},\bH^g_{\omega}),(D\backslash \bar{\Sigma};\epsilon_l,\mu_l,\sigma_l)  \wedge (\Omega\backslash\bar{D};\epsilon_m,\mu_m,0))| \\
\leq C\Big(\varepsilon+\tau^{1/2}\|\bE^g_{\omega}\|_{H({\rm curl},\Omega)}+\tau^{1/2}\|\bH^g_{\omega}\|_{H({\rm curl},\Omega)}\Big),
\end{multline}
where $C$ is positive constant independent of $\varepsilon,\tau,(\bE^g_{\omega},\bH^g_{\omega})$ and $\epsilon_a,\mu_a,\sigma_a$.
\label{invi-cloaking}
\end{thm}
By Theorem~\ref{invi-cloaking}, the cloaking layer $(D\backslash\bar{\Sigma};\epsilon_l,\mu_l,\sigma_l)\wedge(\Omega\backslash\bar{D};\epsilon_m,\mu_m,0)$ makes an arbitrary object located in the cloaked region $\Sigma$ nearly invisible to the wave interrogation by $(\bE^g_{\omega},\bH^g_{\omega})$.

The rest of the paper is organized as follows. In section 2, we introduce the interior transmission eigenvalue problem for inhomogeneous media containing obstacles, and derive its variational form. In section 3, we investigate the spectral property of the interior transmission eigenvalue problem. We prove the discreteness and the existence of the interior transmission eigenvalues. In section 4, we introduce the Maxwell-Herglotz approximation. In section 5, we consider the isotropic invisibility cloaking, and establish the near-invisibility results. This paper is ended with a short discussion.
\section{Interior transmission eigenvalue problem for inhomogeneous media containing obstacles}
\subsection{Physical background of the interior transmission eigenvalue problem}

Let us consider the EM configuration \eqref{device}. We assume that $\sigma_m=0$ and first consider the case that $(D; \epsilon_c,\mu_c,\sigma_c)$ is replaced by a PEC obstacle. Then the scattering problem is described by the following Maxwell system for $(\bE,\bH)\in H(\mbox{curl},\mathbb{R}^3\backslash\bar{D})$
\begin{equation}
\begin{cases}
\curl \bE(x)-i\omega \mu_m\bH(x)=0 \quad &~\mbox{in}~ \Omega\backslash \bar{D},\\
\curl \bH(x)+i\omega \epsilon_m\bE(x)=0 \quad &~\mbox{in}~ \Omega\backslash \bar{D},\\
\curl \bE(x)-i\omega\bH(x)=0 \quad &~\mbox{in} ~\mathbb{R}^3\backslash\bar{\Omega},\\
\curl \bH(x)+i\omega\bE(x)=0 \quad &~\mbox{in} ~\mathbb{R}^3\backslash\bar{\Omega},\\
\nu\times \bE =0 \quad &~\mbox{on}~\partial D,\\
\bE-\bE^i,\bH-\bH^i\mbox{ satisfy the Silver-M$\ddot{\mbox{u}}$ller radiation condition},
\end{cases}
\label{R0}
\end{equation}
where $\nu$ is the unit outward normal vector to boundary $\partial D$. By straightforward manipulations, we can also write the Maxwell system (\ref{R0}) in terms of $\bE\in H^2(\mbox{curl},\mathbb{R}^3\backslash\bar{D})$ as follows,
\begin{equation}
\begin{cases}
\curl\curl \bE(x)-\omega^2 n(x) \bE(x)=0 \quad &~\mbox{in}~ \Omega\backslash \bar{D},\\
\curl\curl \bE(x)-\omega^2 \bE(x)=0 \quad &~\mbox{in} ~\mathbb{R}^3\backslash\bar{\Omega},\\
\nu\times \bE(x) =0 \quad &~\mbox{on}~\partial D,\\
\bE-\bE^i\mbox{ satisfies the Silver-M$\ddot{\mbox{u}}$ller radiation condition},
\end{cases}
\label{R}
\end{equation}
where $n(x):=\mu_m(x)\epsilon_m(x)>0$. Let $\bE^s=\bE-\bE^i$ and $\bE^{\infty}$ signify the corresponding scattered wave field and the far-field pattern, respectively. If perfect invisibility is achieved for the scattering system (\ref{R}), namely, $\bE_{\infty}(\hat{x},(\bE^i,\bH^i),(\Omega;\epsilon,\mu,\sigma))\equiv 0,$ then one has that
\beq
\bE^s(x)=0~~\mbox{for}~x\in \mathbb{R}^3\backslash \bar{\Omega}.
\label{es}
\eeq
By using the standard transmission conditions across $\partial \Omega$ for the solution $\bE$ and $\bH$, one has
\beq
\nu\times \bE|_-=\nu\times \bE|_+~~\mbox{and}~~\nu\times \bH|_-=\nu\times\bH|_+~~\mbox{on}~\partial\Omega,
\label{tran}
\eeq
where
$$\nu\times\bE|_{\pm}(x):=\lim_{h\rightarrow +0}\nu\times\bE(x\pm h\nu(x)),~~x\in\partial\Omega,$$
$$\nu\times\bH|_{\pm}(x):=\lim_{h\rightarrow +0}\nu\times\bH(x\pm h\nu(x)),~~x\in\partial\Omega.$$
The second condition in (\ref{tran}) can also be written as 
$$\nu\times \mu_m^{-1} \curl \bE|_-=\nu\times\curl \bE|_+~~\mbox{on}~\partial\Omega.$$
We assume that the the magnetic permeability $\mu_m=\mathbf{I}_{3 \times 3}$ in $\Omega\backslash\bar{D}$. The above condition becomes
$$\nu\times\curl \bE|_-=\nu\times\curl \bE|_+~~\mbox{on}~\partial\Omega.$$

Applying (\ref{es}) to (\ref{tran}) and by setting $\bE_m(x):=\bE(x)$ for $x\in \Omega\backslash\bar{D}$ and $\bE_0(x):=\bE^i(x)$ for $x\in\Omega$, one can readily show that if perfect invisibility is achieved, then the following interior transmission eigenvalue problem
\begin{equation}
\begin{cases}
\curl\curl \bE_m-\omega^2n(x) \bE_m=0 \quad &~\mbox{in}~ \Omega\backslash \bar{D},\\
\curl\curl \bE_0-\omega^2 \bE_0=0 \quad &~\mbox{in} ~\Omega,\\
\nu\times \bE_m =0 \quad &~\mbox{on}~\partial D,\\
\nu\times \bE_m=\nu\times \bE_0,~~\nu\times\curl\bE_m=\nu\times\curl\bE_0 \quad &~\mbox{on} ~\partial\Omega
\end{cases}
\label{u}
\end{equation} 
has nontrivial solutions, where $n(x)=\epsilon_m$. Note that the interior transmission eigenvalue problem (\ref{u}) is non-self adjoint. The equivalent interior transmission eigenvalue system of (\ref{u}) for the pair $(\bE,\bH)$ is (\ref{u0}).


\subsection{Variational formulation of the interior transmission eigenvalue problem}
In this subsection we are going to derive the variational form of the interior transmission problem (\ref{u}).

Let $\tilde{\bE}=:\bE_m-\bE_0$ in $\Omega\backslash \bar{D}$. By the first and second identities in (\ref{u}) we have that $\tbE$ satisfies

\begin{align}  \label{eq:1}
\curl\curl\tbE-\omega^2n\tbE=\omega^2(n-1)\bE_0~\mbox{in}~\Omega\backslash \bar{D}.
\end{align}

We also get the boundary conditions
$$\nu\times \tbE =0,~~\nu\times\curl \tbE=0~\mbox{on}~\partial\Omega,$$
and
$$
\nu\times \tbE=-\nu\times \bE_0~\mbox{on}~ \partial D.
$$
Then the interior transmission problem (\ref{u}) can be reformulated in terms of $\tbE$ and $\bE_0$ as follows:
\begin{equation*}
\begin{cases}
\curl\curl\tbE-\omega^2n \tbE=\omega^2(n-1)\bE_0 \quad &~\mbox{in}~ \Omega\backslash \bar{D},\\
\curl\curl\bE_0-\omega^2 \bE_0=0 \quad &~\mbox{in} ~\Omega,\\
\nu\times \tbE=-\nu\times \bE_0 \quad &~\mbox{on}~\partial D,\\
\nu\times \tbE=0,~\nu\times \curl\tbE=0 \quad &~\mbox{on} ~\partial\Omega,
\end{cases}
\end{equation*}
With continuity of the data $\nu\times \bE_0$ and $\nu\times\curl\bE_0$ across $\partial D$, we have from (\ref{eq:1}) that
\beq
\nu\times\Big(\frac{1}{\omega^2}(n-1)^{-1}(\mbox{curl curl}-\omega^2n) \tbE\Big)\Big|_+=\nu\times \bE_0|_- ~~\mbox{on}~ \partial D,
\label{e1}
\eeq
and
\beq \nu\times\curl\Big(\frac{1}{\omega^2}(n-1)^{-1}(\mbox{curl curl}-\omega^2n) \tbE\Big)\Big|_+=\nu\times \curl\bE_0|_{-} ~~\mbox{on}~ \partial D.
\label{e2}
\eeq

Multiplying both sides of (\ref{eq:1}) by $(n-1)^{-1}$ and applying operator $(\mbox{curl curl} -\omega^2)$ on both sides, we get a fourth order equation for $\tbE$ in $\Omega\backslash \bar{D}$:
$$
(\mbox{curl curl}-\omega^2)(n-1)^{-1}(\mbox{curl curl}-\omega^2n)\tbE=0~\mbox{in}~\Omega\backslash \bar{D}.
$$
Note that $\tbE$ is only defined in $\Omega\backslash \bar{D}$. We define
\begin{equation*}
\bE=
\begin{cases}
\tbE~~ &\mbox{in}~\Omega\backslash \bar{D},\\
-\bE_0~~&\mbox{in}~D.
\end{cases}
\label{E}
\end{equation*}
Then $\bE\in W$, where $W$ is defined by
$$W:=\left \{\bE \in H_0(\mbox{curl},\Omega): \curl\curl\bE|_{\Omega\backslash\bar{D}}\in L^2(\Omega\backslash \bar{D})\right \},$$
which is equipped with the norm
\begin{equation}
\|\bE\|^2_W=\|\bE\|^2_{H(\mbox{curl},\Omega)}+\|\curl\curl\bE\|^2_{L^2(\Omega\backslash \bar{D})}.
\label{Wnorm}
\end{equation}
Now we are able to formulate a fourth order system of $\bE\in W$ as follows,
\begin{equation}
\begin{cases}
(\mbox{curl curl} -\omega^2)(n-1)^{-1}(\mbox{curl curl}-\omega^2n)\bE=0 \quad &~\mbox{in}~ \Omega\backslash \bar{D},\\
\curl\curl\bE-\omega^2 \bE=0 \quad &~\mbox{in} ~D,\\
\nu\times \bE|_+=\nu\times \bE|_- \quad &~\mbox{on}~\partial D,\\
\nu\times\left(\dfrac{1}{\omega^2}(n-1)^{-1}(\mbox{curl curl}-\omega^2n) \bE\right)\Big|_+=-\nu\times \bE|_- \quad &~\mbox{on}~\partial D,\\
\nu\times\curl\left(\dfrac{1}{\omega^2}(n-1)^{-1}(\mbox{curl curl}-\omega^2n) \bE\right)\Big|_+=-\nu\times \curl\bE|_- \quad &~\mbox{on}~\partial D,\\
\nu\times \bE=0,~\nu\times \curl\bE=0 \quad &~\mbox{on} ~\partial\Omega.
\end{cases}
\label{u2}
\end{equation}

Take a test vector function $\bPhi\in W$. Multiplying the first equation in (\ref{u2}) by $\bar{\bPhi}$ we have
\beq
0=\int_{\Omega\backslash \bar{D}}(\mbox{curl curl} -\omega^2)(n-1)^{-1}(\mbox{curl curl}-\omega^2n)\bE\cdot \bar{\bPhi}\diff x.
\label{Ep}
\eeq
Denote
\beq
\bPsi=(n-1)^{-1}(\mbox{curl curl}-\omega^2n)\bE~~\mbox{in} ~\Omega\backslash \bar{D}.
\label{ksi}
\eeq

With the aid of the vector identity
\beq
\curl\curl\bE=-\Delta \bE+\nabla\div\bE,
\label{ccl}
\eeq
by Green's second vector theorem, together with the boundary conditions in (\ref{u2}), equation (\ref{Ep}) becomes
\beq
\begin{split}
0&=\int_{\Omega\backslash \bar{D}}\bPsi\cdot\overline{(\mbox{curl curl} -\omega^2)\bPhi} \diff x-\int_{\partial D}\big(\nu\times\bPsi\cdot\overline{\curl\bPhi}-\curl\bPsi\cdot\overline{\nu\times\bPhi}\big)\diff s.
\end{split}
\label{u4}
\eeq

By (\ref{ksi}), the fourth and fifth transmission conditions in (\ref{u2}) can be written in terms of
$$\nu\times\bPsi|_+=-\omega^2\nu\times \bE|_-,~~\nu\times\curl\bPsi|_+=-\omega^2\nu\times \curl\bE|_-~~\mbox{on}~\partial D.$$
Use the above boundary conditions, (\ref{u4}) becomes
\begin{multline*} 
0=\int_{\Omega\backslash \bar{D}}(n-1)^{-1}(\curl\curl\bE-\omega^2\bE)\cdot \overline{(\curl\curl\bPhi -\omega^2\bPhi})\diff x-\int_{\Omega\backslash \bar{D}}\omega^2 \bE\cdot  \overline{\curl\curl\bPhi}\diff x\\
+\int_{\Omega\backslash \bar{D}}\omega^4\bE\cdot \bar{\bPhi}\diff x+\omega^2\int_{\partial D}\big(\nu\times \bE|_- \cdot \overline{\curl\bPhi}+\nu\times\curl\bE|_-\cdot\bar{\bPhi}\big)\diff s.
\end{multline*}
By Green's first vector theorem, and with the aid of the vector identity (\ref{ccl}), we obtain
\begin{multline*}
0 =\int_{\Omega\backslash \bar{D}}(n-1)^{-1}(\curl\curl\bE-\omega^2\bE)\cdot \overline{(\curl\curl\bPhi -\omega^2\bPhi})\diff x \\
-\omega^2\int_{\Omega}\curl\bE\cdot \overline{\curl\bPhi}\diff x+\omega^4\int_{\Omega}\bE\cdot \bar{\bPhi}\diff x.
\end{multline*}

Therefore, the variational formulation of the interior transmission problem (\ref{u2}) becomes:
find $\bE\in W$ such that
\begin{multline}
\int_{\Omega\backslash \bar{D}}(n-1)^{-1}(\curl\curl\bE-\omega^2\bE)\cdot \overline{(\curl\curl\bPhi -\omega^2\bPhi})\diff x \\
-\omega^2\int_{\Omega}\curl\bE\cdot \overline{\curl\bPhi}\diff x+\omega^4\int_{\Omega}\bE\cdot \bar{\bPhi}\diff x=0,
\label{var}
\end{multline}
for all $\bPhi\in W$. By taking appropriate test function it is easy to see that a solution of the variational problem (\ref{var}) defines a week solution to (\ref{u2}) and therefore to the interior transmission problem (\ref{u}).

\section{Spectral property of the interior transmission eigenvalue problem}
In this section, we investigate the spectral property of the interior transmission eigenvalue problem (\ref{u}). We prove the discreteness and the existence of the interior transmission eigenvalues by considering $n-1<0$ in $\Omega\backslash\bar{D}$ and $n-1>0$ in $\Omega\backslash\bar{D}$, respectively.
\subsection{Discreteness of the spectrum when $n-1<0$}
\begin{thm}
Assume that $n(x)-1<0$ in $\Omega\backslash\bar{D}$. And moreover, $n(x)=\epsilon_m(x)$ satisfies $\|\nabla n(x)/n(x)\|_{L^{\infty}(\Omega\backslash \bar{D})}\ll 1$. Then the set of transmission eigenvalues is discrete.
\label{negative}
\end{thm}
\begin{proof}
Suppose $n(x)-1<0$ in $\Omega\backslash\bar{D}$. We define two sesquilinear forms on $W\times W$:
\begin{multline*}
\mathcal{A}_\omega(\bE,\bPhi)=-\int_{\Omega\backslash \bar{D}}(n-1)^{-1}(\curl\curl\bE-\omega^2\bE)\cdot \overline{(\curl\curl\bPhi -\omega^2\bPhi})\diff x\\
+\omega^2\int_{\Omega}\curl\bE\cdot \overline{\curl\bPhi}\diff x+\omega^4\int_{\Omega}\bE\cdot \bar{\bPhi}\diff x
\end{multline*}
and
$$\mathcal{B}(\bE,\bPhi)=2\int_{\Omega}\bE\cdot \bar{\bPhi}\diff x,$$
where $W$ is defined by (\ref{Wnorm}). Then the variational form (\ref{var}) of the interior transmission is equivalent to
$$ \mbox{Find }~\bE\in W~ \mbox{such that}$$
$$\mathcal{A}_\omega(\bE,\bPhi)-\omega^4\mathcal{B}(\bE,\bPhi)=0~~\mbox{for all}~\bPhi\in W.$$
By the Riesz representation theorem there exist two bounded linear operators $\mathbf{A}_\omega:W\rightarrow W$ and $\bB:W\rightarrow W$ such that
$$(\bA_\omega \bE,\bPhi)_W:=\mathcal{A}_\omega(\bE,\bPhi)~\mbox{and}~(\bB \bE,\bPhi)_W:=\mathcal{B}(\bE,\bPhi).$$

By Lemma \ref{Acoercive} and \ref{Bcompact} in the sequel we have $\bA_\omega$ is coercive and $\bB$ is compact. Hence the operator $\bA_\omega -\omega ^4\bB $ is Fredholm with index zero. The transmission eigenvalues are the values of $\omega >0$ for which $\bI-\omega ^4\bA_\omega ^{-1}\bB $ has a nontrivial kernel. To apply the analytic Fredholm theorem, it remains to show that $\bI-\omega ^4\bA_\omega ^{-1}\bB $ or $\bA_\omega -\omega ^4\bB $ is injective for at least one $\omega $.

For all $\bE\in W$ we have that
\begin{multline}\label{AB}
\mathcal{A}_\omega (\bE,\bE)-\omega ^4\mathcal{B} (\bE,\bE)\\
=\int_{\Omega\backslash \bar{D}}(1-n)^{-1}|\curl\curl \bE-\omega ^2\bE|^2\diff x+\omega ^2\|\curl \bE\|^2_{L^2(\Omega)}-\omega^4\|\bE\|^2_{L^2(\Omega)}.
\end{multline}
The Poincar\'{e} inequality gives us that:
\begin{equation}
\|\bE\|^2_{L^2(\Omega)}\leq C\left(\|\curl \bE\|^2_{L^2(\Omega)}+\|\div \bE\|^2_{L^2(\Omega)}\right),
\label{poincare}
\end{equation}
for all $\bE$ satisfies $\nu\times \bE=0$ on $\partial\Omega$, where constant $C$ is independent of $\bE$ (see \cite{Monk}).

We observe from (\ref{u}) that $\div \bE=0$ in $D$, then, for all $\bE\in W$ we have from (\ref{AB}) that
\beq
\mathcal{A}_\omega (\bE,\bE)-\omega ^4\mathcal{B}(\bE,\bE)\geq\omega ^2(1-\omega^2C)\|\curl \bE\|^2_{L^2(\Omega)}-\omega ^4C\|\div\bE\|^2_{L^2(\Omega\backslash \bar{D})}.
\label{curlo}
\eeq
Since $\bE=\bE_m-\bE_0$ and $\div\bE_0=0$ in $\Omega\backslash \bar{D}$, we have $\div \bE=\div \bE_m$. From (\ref{u}) we can see that $\div (n\bE_m)=0$ in $\Omega\backslash \bar{D}$. Due to the identity $\div  (n\bE_m)=n(\div \bE_m)+\nabla n\cdot\bE_m$ in $\Omega\backslash \bar{D}$, we have
$$\|\div \bE\|_{L^2(\Omega\backslash \bar{D})}=\|\div \bE_m\|_{L^2(\Omega\backslash \bar{D})}\leq \|\nabla n(x)/n(x)\|_{L^{\infty}(\Omega\backslash \bar{D})}\|\bE_m\|_{L^2(\Omega\backslash \bar{D})}.$$
If $\|\nabla n(x)/n(x)\|_{L^2(\Omega\backslash \bar{D})}\ll 1$, then it follows from (\ref{curlo}) that $\mathcal{A}_\omega (\bE,\bE)-\omega ^2\mathcal{B} (\bE,\bE)>0$ for all $\omega >0$ such that $\omega<\sqrt{1/C}$. Hence $\mathcal{A}_\omega -\omega ^2\mathcal{B} $ is injective for such $\omega $ and the analytical Fredholm theorem implies that the set of transmission eigenvalues is discrete.
\end{proof}

\begin{lem}
The operator $\bA_\omega$ is coercive. 
\label{Acoercive}
\end{lem}
\begin{proof}Taking specific $\bPhi=\bE\in W$, we have
\begin{multline}\label{A}
(\bA_\omega \bE,\bE)_W=\int_{\Omega\backslash \bar{D}}(1-n)^{-1}|\curl\curl \bE-\omega^2\bE|^2\diff x+\omega^2\|\curl \bE\|^2_{L^2(\Omega)}+\omega^4\|\bE\|^2_{L^2(\Omega)}\\
=\int_{\Omega\backslash \bar{D}}(1-n)^{-1}\big(|\curl\curl \bE|^2-2\mbox{Re}\{\omega^2\bE\cdot\curl\curl \bE\}+\omega^4|\bE|^2\big)\diff x\\
+\omega^2\|\curl \bE\|^2_{L^2(\Omega)}+\omega^4\|\bE\|^2_{L^2(\Omega)}.
\end{multline}

Denote $n_{-}=\inf_{x\in\Omega\backslash \bar{D}}n(x)$, and $n^{+}=\sup_{x\in\Omega\backslash \bar{D}}n(x)$. We assume $0<n_-<n(x)<n^+<1$ for $x\in\Omega\backslash \bar{D}$. Set $\gamma=\frac{1}{1-n_-}$. Using the equality
\beq
\gamma X^2-2\gamma XY+(1+\gamma)Y^2=\delta \left( Y-\frac{\gamma}{\delta}X \right)^2+ \left(\gamma-\frac{\gamma^2}{\delta}\right)X^2+(1+\gamma-\delta)Y^2,
\label{ineq}
\eeq
for $X=\|\curl\curl \bE\|_{L^2(\Omega\backslash \bar{D})}$, $Y=\omega^2\|\bE\|_{L^2(\Omega\backslash \bar{D})}$ and arbitrary $\delta>0$, (\ref{A}) becomes
\begin{multline*}
(\bA_\omega \bE,\bE)_W\geq\gamma\|\curl\curl \bE\|^2_{L^2(\Omega\backslash \bar{D})}-2\omega^2\gamma\|\curl\curl \bE\|_{L^2(\Omega\backslash \bar{D})} \|\bE\|_{L^2(\Omega\backslash \bar{D})}\\
+\omega^4(1+\gamma)\|\bE\|^2_{L^2(\Omega\backslash \bar{D})}
+\omega^2\|\curl\bE\|^2_{L^2(\Omega)}+\omega^4\|\bE\|^2_{L^2(D)}\\
\geq \left( \gamma-\frac{\gamma^2}{\delta} \right)\|\curl\curl \bE\|^2_{L^2(\Omega\backslash \bar{D})}+\omega^4(1+\gamma-\delta)\|\bE\|^2_{L^2(\Omega\backslash \bar{D})}+\omega^2\|\curl \bE\|^2_{L^2(\Omega)}+\omega^4\|\bE\|^2_{L^2(D)},
\end{multline*}
where $\gamma<\delta<\gamma+1$. For such an $\delta$, we conclude that there exists a constant $C>0$ such that
$$(\bA_\omega \bE,\bE)_W\geq C\|\bE\|^2_W$$
for all $\bE\in W$ which proves that $\bA_\omega:W\rightarrow W$ is coercive.
\end{proof}

\begin{lem}
The operator $\bB$ is compact on $W$. 
\label{Bcompact}
\end{lem}
\begin{proof}
From the Cauchy-Schwarz inequality we have that
$$|(\bB \bE,\bPhi)_W|\leq 2\|\bE\|_{L^2(\Omega)}\|\bPhi\|_{L^2(\Omega)}.$$
As a consequence, we have 
$$\|\bB \bE\|_W=\sup_{\bPhi\in W, \bPhi\neq 0}\frac{|(\bB \bE,\bPhi)_W|}{\|\bPhi\|_{W}}\leq 2\|\bE\|_{L^2(\Omega)},$$
for every $\bE\in W$. Thus $\bB$ appears as a continuous operator from $L^2(\Omega)$ into $W$.  It is easy to check that under the assumption that $\|\nabla n(x)/n(x)\|_{L^2(\Omega\backslash \bar{D})}\ll 1$, $W$ is a subspace of $H(\mbox{curl},\div n;\Omega)$, see B.1 in \cite{HL}, where the compactness of embedding of $H(\mbox{curl},\div n;\Omega)$(under boundary condition $\nu\times\bE=0$ on $\partial\Omega$) into $L^2(\Omega)$ has been shown to be compact.
\end{proof}

From the proof of the previous theorem, we deduce a lower bound for the first transmission eigenvalue when $n-1$ is negative in $\Omega\backslash\bar{D}$, that is, $\omega \geq \sqrt{1/C}$, where $C$ is the constant in \eqref{poincare}.

\subsection{Existence of the interior transmission eigenvalues when $n-1<0$}
In this subsection, we prove the existence of transmission eigenvalues following \cite{CCH,CGH} where the acoustic wave was considered.

If we consider the generalized eigenvalue problem
\beq
\bA_\omega \bE-\lambda(\omega )\bB\bE=0,~~\bE\in W,
\label{eig}
\eeq
which is known to have an infinite sequence of eigenvalues $\lambda_j(\omega )$, $j\in\mathbb{N}$. The proof of the existence of transmission eigenvalues makes use of the following theorem shown in \cite{CH}.
\begin{thm}
Let $\tau\rightarrow \bA_\tau$ be a continuous mapping from $(0,\infty)$ to the set of self-adjoint and positive definite bounded linear operator on $W$ and let $\bB$ be a self-adjoint and non-negative compact bounded linear operator on $W$. We assume that there exists two positive constants $\tau_0>0$ and $\tau_1>0$ such that
\begin{enumerate}
\item $\bA_{\tau_0}-\tau_0\bB$ is positive on $W$,
\item $\bA_{\tau_1}-\tau_1\bB$ is non-positive on a m-dimensional subspace of $W$.
\end{enumerate}
Then each of the equations $\lambda_j(\tau)=\tau$ for $j=1,\dots,m$, has at least one solution in $[\tau_0,\tau_1]$ where $\lambda_j(\tau)$ is the $j^{th}$ eigenvalue (counting multiplicity) of $\bA_\omega $ with respect to $\bB$, {\it i.e.}, ker\,$(\bA_\tau-\lambda_j(\tau)\bB)\neq \{0\}$.
\label{exi}
\end{thm}

Suppose $n-1<0$ in $\Omega\backslash \bar{D}$, then the transmission eigenvalues are the solutions $\lambda_j(\omega )=\omega ^4$ of (\ref{eig}), $j\in\mathbb{N}$. The existence of the interior eigenvalues is given in the following theorem.
\begin{thm}
If $n-1<0$ in $D$, then there exists an infinite discrete set of transmission eigenvalues.
\label{exim}
\end{thm}
\begin{proof}
In the proof of Theorem \ref{negative} we have shown that for $\omega<\sqrt{1/C}$, where $C$ is given in (\ref{poincare}), $\bA_{\omega _0}-\omega ^4_0\bB$ is positive in $W$. Hence, the first assumption in Theorem \ref{exi} is satisfied. Now we try to find $\omega _1$ such that the second assumption is also satisfied.

Let $B_r^j$, $j=1,\dots,M(r)$, be $M(r)$ disjoint balls of radius $r$ included in $\Omega\backslash \bar{D}$. We denote by $\omega _1$ the first transmission eigenvalue corresponding to the interior transmission problem for $B_r^j$ for all $j=1,\dots,M(r)$ with index of refraction $n^+$ which is known to exist \cite{CK}, where $n^{+}:=\sup_{x\in\Omega\backslash \bar{D}}n(x)$. Let $\bE_j\in H^2_0(\mbox{curl},B_r^j)$, $j=1,\dots,M(r)$, be the corresponding eigenvector which satisfies
\beq
\int_{B_r^j}\frac{1}{1-n^+}(\curl\curl \bE_j-\omega _1^2n^+\bE_j)\cdot\overline{(\curl\curl \bPhi-\omega _1^2\bPhi})\diff x=0
\label{exm0}
\eeq
for all $\bPhi\in H^2_0(\mbox{curl},B_r^j)$. We denote by $\bE^0_j\in H^2_0(\mbox{curl},\Omega)$ the extension of $\bE_j$ by zero to the whole of $\Omega$ and we define a $M(r)$-dimensional subspace of $W$ by $V:=\mbox{span}\{\bE^0_j,1\leq j\leq M(r)\}$. Since for $j\neq m$, $\bE^0_j$ and $\bE^0_m$ have disjoint support, for $\bE=\sum_{j=1}^{j=M(r)}\alpha_j\bE^0_j\in V$, we have
\begin{multline}
\mathcal{A}_{\omega _1}(\bE,\bE)-\omega _1^2\mathcal{B}(\bE,\bE)\\
=\sum_{j=1}^{M(r)}|\alpha_j|^2\Big(\int_{\Omega\backslash \bar{D}}(1-n)^{-1}|\curl\curl \bE^0_j-\omega _1^2\bE^0_j|^2\diff x+\omega _1^2\int_{\Omega}|\curl\bE^0_j|^2\diff x-\omega _1^4\int_{\Omega}|\bE^0_j|^2\diff x\Big)\\
=\sum_{j=1}^{M(r)}|\alpha_j|^2\Big(\int_{B_r^j}(1-n)^{-1}|\curl\curl \bE_j-\omega _1^2\bE_j|^2\diff x+\omega _1^2\int_{B_r^j}|\curl \bE_j|^2\diff x-\omega _1^4\int_{B_r^j}|\bE_j|^2\diff x\Big)\\
\leq \sum_{j=1}^{M(r)}|\alpha_j|^2\Big(\frac{1}{1-n^+}\int_{B_r^j}|\curl\curl \bE_j-\omega _1^2\bE_j|^2\diff x+\omega _1^2\int_{B_r^j}|\curl \bE_j|^2\diff x-\omega _1^4\int_{B_r^j}|\bE_j|^2\diff x\Big).
\label{ns1}
\end{multline}
By Green's first vector theorem, we have
\begin{equation*}
\mbox{Re}
\left \{\int_{B_r^j}\curl\curl \bE_j\cdot\bar{\bE}_j\diff x\right \}
=\int_{B_r^j}|\curl \bE_j|^2\diff x-\mbox{Re}
\left \{\int_{\partial B_r^j}\curl\bE_j\cdot\overline{\nu\times \bE_j}\diff s\right \}.
\end{equation*}
Since $\nu\times\bE_j=0$ on $\partial B_r^j$, then following equality holds
\beq
\int_{B_r^j}|\curl \bE_j|^2\diff x=\mbox{Re}
\left \{\int_{B_r^j}\curl\curl \bE_j\cdot\bar{\bE}_j\diff x\right \}.
\label{ex}
\eeq
Substituting (\ref{ex}) into (\ref{ns1}), together with (\ref{exm0}), we have
\begin{multline*}
\mathcal{A}_{\omega _1}(\bE,\bE)-\omega _1^2\mathcal{B}(\bE,\bE)\\
\leq\sum_{j=1}^{M(r)}|\alpha_j|^2\Big(\int_{B_r^j}\frac{1}{1-n^+}(\curl\curl \bE_j-\omega _1^2n^+\bE_j)\cdot\overline{(\curl\curl\bE_j-\omega _1^2\bE_j)}\diff x\Big)=0.
\end{multline*}
Hence, the second assumption in Theorem \ref{exi} is satisfied.

Thus, we conclude that there exist $M(r)$ transmission eigenvalues in $(\sqrt{1/C},\omega _1]$. Letting $r\rightarrow 0$, we have that $M(r)\rightarrow \infty$ and thus we can now deduce that there exists an infinite set of transmission eigenvalues.

\end{proof}

\subsection{Discreteness of the spectrum when $n-1>0$}
\begin{thm}
Assume that $n(x)-1>0$ for $x\in \Omega\backslash \bar{D}$. Then the set of transmission eigenvalues is discrete.
\label{positive}
\end{thm}
\begin{proof}
Suppose $n(x)-1>0$ for $x\in \Omega\backslash \bar{D}$. Define sesquilinear forms
\begin{multline*}
\mathcal{F}_\omega (\bE,\bPhi)=\int_{\Omega\backslash \bar{D}}(n-1)^{-1}(\curl\curl \bE-\omega ^2\bE)\cdot \overline{(\curl\curl\bPhi -\omega ^2\bPhi})\diff x\\
+\omega ^2\int_{\Omega}\curl \bE\cdot \overline{\curl\bPhi}\diff x+\omega ^4\int_{\Omega}\bE\cdot \bar{\bPhi}\diff x,
\end{multline*}
and
$$\mathcal{G}_\omega (\bE,\bPhi)=2\int_{\Omega}\curl \bE\cdot \overline{\curl\bPhi}\diff x.$$
Then the variational form (\ref{var}) of the interior transmission problem becomes: find $\bE\in W$ such that
$$\mathcal{F}_\omega (\bE,\bPhi)-\omega ^2\mathcal{G} (\bE,\bPhi)=0~~\mbox{for all}~\bPhi\in W.$$
By the Riesz representation theorem there exist two bounded linear operators $\bF_\omega :W\rightarrow W$ and $\bG :W\rightarrow W$ such that
$$(\bF_\omega \bE,\bPhi)_W:=\mathcal{F}_\omega (\bE,\bPhi)~\mbox{and}~(\bG \bE,\bPhi)_W:=\mathcal{G} (\bE,\bPhi).$$

The operator $\bG $ is compact which has been shown in (Lemma 3.4, \cite{H}). Suppose $\bF_\omega $ is coercive, whose proof will be given in lemma \ref{Fcoercive}.  Then the operator $\bF_\omega -\omega ^2\bG $ is Fredholm with index zero. The transmission eigenvalues are the values of $\omega >0$ for which $\bI-\omega ^2\bF_\omega ^{-1}\bG $ has a nontrivial kernel. To apply the analytic Fredholm theorem, it remains to show that $\bI-\omega ^2\bF_\omega ^{-1}\bG $ or $\bF_\omega -\omega ^2\bG $ is injective for at least one $\omega $.

For all $\bE\in W$ we have that
\begin{multline}\label{n1}
\mathcal{F}_\omega (\bE,\bE)-\omega ^2\mathcal{G} (\bE,\bE)=\int_{\Omega\backslash \bar{D}}(n-1)^{-1}\Big(|\curl\curl \bE|^2-2\mbox{Re}\left\{\omega ^2\bE\cdot\curl\curl \bE\right\}+\omega ^4|\bE|^2\Big)\diff x\\
-\omega ^2\|\curl \bE\|^2_{L^2(\Omega)}+\omega ^4\|\bE\|^2_{L^2(\Omega)}.
\end{multline}

The Poincar\'{e} inequality (\ref{poincare}) of $\curl\bE$ gives us that:
\begin{equation} \label{eq:2}
\|\curl\bE\|^2_{L^2(\Omega)}\leq C_1\|\curl\curl\bE\|^2_{L^2(\Omega)}
\end{equation}
for all $\bE$ satisfies $\nu\times\curl \bE=0$ on $\partial\Omega$, where constant $C_1$ is independent of $\bE$. Since $\curl\curl \bE=\omega^2\bE$ in $D$ from (\ref{u2}), then (\ref{n1}) becomes
\begin{multline}\label{n2}
\mathcal{F}_\omega (\bE,\bE)-\omega ^2\mathcal{G} (\bE,\bE)
\geq\int_{\Omega\backslash \bar{D}}(n-1)^{-1}\Big(|\curl\curl \bE|^2-2\mbox{Re}\left\{\omega ^2\bE\cdot\curl\curl \bE\right\}+\omega ^4|\bE|^2\Big)\diff x\\
-\omega ^2C_1\|\curl\curl\bE\|^2_{L^2(\Omega)}+\omega ^4\|\bE\|^2_{L^2(\Omega)}\\
\geq\int_{\Omega\backslash \bar{D}}(n-1)^{-1}\Big(|\curl\curl \bE|^2-2\mbox{Re}\left\{\omega ^2\bE\cdot\curl\curl \bE\right\}+\omega ^4|\bE|^2\Big)\diff x\\
-\omega ^2C_1\|\curl\curl\bE\|^2_{L^2(\Omega\backslash\bar{D})}+\omega ^4(1-C_1\omega^2 )\|\bE\|^2_{L^2(D)}+\omega ^4\|\bE\|^2_{L^2(\Omega\backslash\bar{D})}.
\end{multline}

Denote again  $\gamma=\frac{1}{n^+-1}$, then (\ref{n2}) becomes
\begin{multline}\label{n3}
\mathcal{F}_\omega (\bE,\bE)-\omega ^2\mathcal{G} (\bE,\bE)
\geq \int_{\Omega\backslash \bar{D}}\gamma\Big((1-C_1\omega^2/\gamma)|\curl\curl \bE|^2-2\mbox{Re}\left\{\omega ^2\bE\cdot\curl\curl \bE\right\}+\omega ^4|\bE|^2\Big)\diff x\\
+\omega ^4(1-C_1\omega^2 )\|\bE\|^2_{L^2(D)}+\omega ^4\|\bE\|^2_{L^2(\Omega\backslash\bar{D})}.
\end{multline}

Suppose $1-C_1\omega^2/\gamma\geq 0$, then (\ref{n3}) becomes
\begin{multline*}
\mathcal{F}_\omega (\bE,\bE)-\omega ^2\mathcal{G} (\bE,\bE)\geq\int_{\Omega\backslash \bar{D}}\gamma\Big|\sqrt{1-C_1\omega^2/\gamma}\curl\curl \bE-\frac{1}{\sqrt{1-C_1\omega^2/\gamma}}\omega^2\bE\Big|^2\diff x\\
+(1-C_1\omega^2(1+1/\gamma))\omega ^4\|\bE\|^2_{L^2(\Omega\backslash\bar{D})}+\omega ^4(1-C_1\omega^2 )\|\bE\|^2_{L^2(D)}.
\end{multline*}
Therefore, we deduce that $\mathcal{F}_\omega (\bE,\bE)-\omega ^2\mathcal{G} (\bE,\bE)>0$ under the condition that
$$0\leq\omega^2\leq \frac{\gamma}{\gamma+1}\frac{1}{C_1},$$
and hence $\mathcal{F}_\omega -\omega ^2\mathcal{G} $ is injective. Hence, the analytical Fredholm theorem implies that the set of transmission eigenvalues is discrete and from the analyticity with $+\infty$ the only possible accumulating point.
\end{proof}

\begin{lem}
$\bF_\omega$ is coercive.
\label{Fcoercive}
\end{lem}
\begin{proof}
Taking $\bPhi=\bE\in W$, we have
\begin{multline}\label{F}
(\bF_\omega \bE,\bE)_W=\int_{\Omega\backslash \bar{D}}(n-1)^{-1}|\curl\curl \bE-\omega ^2\bE|^2\diff x+\omega ^2\|\curl \bE\|^2_{L^2(\Omega)}+\omega ^4\|\bE\|^2_{L^2(\Omega)}\\
=\int_{\Omega\backslash \bar{D}}(n-1)^{-1}\big(|\curl\curl \bE|^2-2\mbox{Re}\left\{\omega ^2\bE\cdot\curl\curl \bE\right\}+\omega ^4|\bE|^2\big)\diff x\\
+\omega ^2\|\curl \bE\|^2_{L^2(\Omega)}+\omega ^4\|\bE\|^2_{L^2(\Omega)}.
\end{multline}

Denote $n_{-}=\inf_{\Omega\backslash \bar{D}}n(x)$, and $n^{+}=\sup_{\Omega\backslash \bar{D}}n(x)$. We assume $n^+>n(x)>n_->1$ for $x\in\Omega\backslash \bar{D}$. Set $\gamma=\frac{1}{n^+-1}$. Using inequality (\ref{ineq}) we have from (\ref{F}) that
\begin{multline*}
(\bF_\omega \bE,\bE)_W\geq\gamma\|\curl\curl \bE\|^2_{L^2(\Omega\backslash \bar{D})}-2\omega ^2\gamma\|\curl\curl \bE\|_{L^2(\Omega\backslash \bar{D})} \|\bE\|^2_{L^2(\Omega\backslash \bar{D})}\\
+\omega ^4(1+\gamma)\|\bE\|^2_{L^2(\Omega\backslash \bar{D})}+\omega ^2\|\curl \bE\|^2_{L^2(\Omega)}+\omega ^4\|\bE\|_{L^2(D)}\\
\geq \left( \gamma-\frac{\gamma^2}{\delta} \right)\|\curl\curl \bE\|^2_{L^2(\Omega\backslash \bar{D})}+\omega ^4(1+\gamma-\delta)\|\bE\|^2_{L^2(\Omega\backslash \bar{D})}+\omega ^2\|\curl \bE\|^2_{L^2(\Omega)}+\omega ^4\|\bE\|^2_{L^2(D)},
\end{multline*}
where $\gamma<\delta<\gamma+1$. For such an $\delta$, we conclude that there exists a constant $C>0$ such that
$$(\bF_\omega \bE,\bE)_W\geq C\|\bE\|^2_W$$
for all $\bE\in W$ which proves that $\bF_\omega :W\rightarrow W$ is coercive.
\end{proof}

From the previous theorem, we deduce a lower bound for the first transmission eigenvalue when $n-1$ is positive in $\Omega\backslash\bar{D}$, that is, $\omega> \sqrt{\frac{\gamma}{\gamma+1}\frac{1}{C_1}}$ where $C_1$ is the constant given in \eqref{eq:2}.

\subsection{Existence of the interior transmission eigenvalues when $n-1>0$}
Suppose $n-1>0$ in $\Omega\backslash \bar{D}$, then the transmission eigenvalues are the solutions $\lambda_j(\omega )=\omega ^2$ of (\ref{eig}), $j\in\mathbb{N}$. The existence of the interior eigenvalues is given in the following theorem. The proof of the existence of transmission eigenvalues again makes use of Theorem \ref{exi}.
\begin{thm}
Assume that $n(x)>1$ for $x\in\Omega\backslash\bar{D}$. There exists an infinite discrete set of transmission eigenvalues.
\label{exip}
\end{thm}
\begin{proof}
The proof is exactly an analogue of Theorem \ref{exim}. Let $\tau=\omega ^2$. We have seen from Theorem \ref{positive} that if $0\leq\omega^2\leq \frac{\gamma}{\gamma+1}\frac{1}{C_1}$, then $\bF_{\omega _0}-\omega ^2_0\bG$ is positive in $W$. Hence, the first assumption in Theorem \ref{exi} is satisfied. Now we try to find $\omega _1$ such that the second assumption is also satisfied.

Let $B_r^j$, $j=1,\dots,M(r)$, be $M(r)$ balls of radius $r$ included in $\Omega\backslash \bar{D}$. We denote by $\omega _2$ the first transmission eigenvalue corresponding to the interior transmission problem for $B_r^j$ for all $j=1,\dots,M(r)$ with index of refraction $n_-$ which is known to exist \cite{CK}, where $n_{-}:=\inf_{x\in\Omega\backslash \bar{D}}n(x)$. Let $\bE_j\in H^2_{0}(\mbox{curl},B_r^j)$, $j=1,\dots,M(r)$, be the corresponding eigenvector which satisfies
\beq
\int_{B_r^j}\frac{1}{n^--1}(\curl\curl \bE_j-\omega _2^2n_-\bE_j)\cdot\overline{(\curl\curl \bPhi-\omega _2^2\bPhi})\diff x=0,
\label{exp0}
\eeq
for all $\bPhi\in H^2_0(\mbox{curl},B_r^j)$. We denote by $\bE^0_j\in H^2_0(\mbox{curl},\Omega)$ the extension of $\bE_j$ by zero to the whole of $\Omega$ and we define a $M(r)$-dimensional subspace of $W$ by $V:=\mbox{span}\{\bE^0_j,1\leq j\leq M(r)\}$. Since for $j\neq m$, $\bE^0_j$ and $\bE^0_m$ have adjoint support, for $\bE=\sum_{j=1}^{j=M(r)}\beta_j\bE^0_j\in V$, we have
\begin{multline}\label{fg}
\mathcal{F}_{\omega _2}(\bE,\bE)-\omega _2^2\mathcal{G}(\bE,\bE)\\
=\sum_{j=1}^{M(r)}|\beta_j|^2\Big(\int_{\Omega\backslash \bar{D}}(n-1)^{-1}|\curl\curl \bE^0_j-\omega _2^2\bE^0_j|^2\diff x-\omega _2^2\int_{\Omega}|\curl \bE^0_j|^2\diff x+\omega _2^4\int_{\Omega}|\bE^0_j|^2\diff x\Big)\\
=\sum_{j=1}^{M(r)}|\beta_j|^2\Big(\int_{B_r^j}(n-1)^{-1}|\curl\curl \bE_j-\omega _2^2\bE_j|^2\diff x-\omega _2^2\int_{B_r^j}|\curl \bE_j|^2\diff x+\omega _2^4\int_{B_r^j}|\bE_j|^2\diff x\Big)\\
\leq \sum_{j=1}^{M(r)}|\beta_j|^2\Big(\frac{1}{n_--1}\int_{B_r^j}|\curl\curl \bE_j-\omega _2^2\bE_j|^2\diff x-\omega _2^2\int_{B_r^j}|\curl \bE_j|^2\diff x+\omega _2^4\int_{B_r^j}|\bE_j|^2\diff x\Big).
\end{multline}

Substituting (\ref{ex}) into (\ref{fg}), together with (\ref{exp0}), we have
\begin{multline*}
\mathcal{F}_{\omega _2}(\bE,\bE)-\omega _2^2\mathcal{G}(\bE,\bE)\\
\leq\sum_{j=1}^{M(r)}|\beta_j|^2\Big(\int_{B_r^j}\frac{1}{n_--1}(\curl\curl \bE_j-\omega _2^2n_-\bE_j)\cdot\overline{(\curl\curl \bE_j-\omega _2^2\bE_j)}\diff x\Big)=0.
\end{multline*}
Hence, the second assumption in Theorem \ref{exi} is satisfied.

Thus, we conclude that there exist $M(r)$ transmission eigenvalues in $[\sqrt{\frac{\gamma}{\gamma+1}\frac{1}{C_1}},\omega_2)$. Letting $r\rightarrow 0$, we have that $M(r)\rightarrow \infty$ and thus we can now deduce that there exists an infinite set of transmission eigenvalues.
\end{proof}

\section{Interior transmission eigenvalue problem and Maxwell-Herglotz approximation}
In this section, we shall find some important applications of the interior transmission problem (\ref{u0}) to invisibility cloaking. According to the discussion in the introduction, if perfect invisibility is obtained for the scattering system (\ref{R0}), then one has the eigenfunctions for the interior transmission eigenvalue problem (\ref{u0}). However, the converse is not necessarily true unless the eigenfunctions $(\bE_0,\bH_0)$ can be smoothly extended from $\Omega$ to $\mathbb{R}^3$. Nevertheless, we shall show that near invisibility can still be achieved under certain circumstances. To this end, we first need to extend the interior transmission eigenfunctions $(\bE_0,\bH_0)\in H(\mbox{curl},\Omega)\times H(\mbox{curl},\Omega)$ in (\ref{u0}) to the whole space $\mathbb{R}^3$ by the so-called Maxwell-Herglotz approximation to form an incident electric wave field for (\ref{R0}). Define
\beq
\bE^g_\omega(x):=\int_{\mathbb{S}^2}e^{i\omega x\cdot d}g(d)\diff s(d), ~~g\in L^2(\mathbb{S}^2),~~\bH^g_\omega(x):=\frac{1}{i\omega}\curl\bE^g_\omega(x),~x\in\mathbb{R}^3.
\label{Herg}
\eeq
$\bE^g_\omega,\bH^g_\omega$ are called the Maxwell-Herglotz wave functions. We have from \cite{W} the following theorem.
\begin{thm}
Suppose $\Omega$ is a domain of class $C^{0,1}$, and $\omega$ is not a PEC eigenvalue of the Maxwell's equations in $\Omega$. Let $\mathcal{H}_{\omega}$ denote the space of all Maxwell-Herglotz functions of the form (\ref{Herg}) restricted for $x \in \Omega$, and
$$W_\omega(\Omega):=\{(\bE,\bH)\in L^2(\Omega)\times L^2(\Omega):{\rm curl} ~\bE-i\omega\bH=0,~{\rm curl}~ \bH+i\omega \bE=0\}$$
denote the space of all solutions of the Maxwell's equations in $\Omega$. Then $\mathcal{H}_\omega$ is dense in $W_\omega(\Omega)$ with respect to the norm $\|\cdot\|_{H({\rm curl}, \Omega)}$. 
\label{dense}
\end{thm}

Starting from now and throughout the rest of the paper, we assume that $\Omega$ is of class $C^{0,1}$. In the context of the Maxwell system (\ref{u0}), an interior transmission eigenvalue $\omega$ is called a non-scattering ``energy" if the interior transmission eigenfunctions $(\bE_0,\bH_0)$ happen to be a pair of Maxwell-Herglotz fields of the form (\ref{Herg}). It can be easily shown that if $\omega$ is a non-scattering ``energy", and one uses the corresponding pair of eigenfunctions $(\bE_0,\bH_0)$ as the incident electric field to impinge on device $(\Omega;\epsilon,\mu,\sigma)$, then far-field pattern generated will be identically vanishing. Next, we shall show that every interior transmission eigenvalue $\omega$ is a nearly non-scattering ``energy".

Now we present the proof of Proposition~\ref{nearcloaking} of this paper which connects the interior transmission eigenvalue problem (\ref{u0}) to the invisibility cloaking.

\begin{proof}[Proof of Proposition~\ref{nearcloaking}]
Since $\omega\in\mathbb{R}^+$ is an interior transmission eigenvalue associated with $(\Omega\backslash\bar{D};\epsilon_m,\mu_m,0)$ and $(\bE_m,\bH_m)$, $(\bE_0,\bH_0)$ are the corresponding eigenfunctions, we see from (\ref{u}) that
\begin{equation}
\begin{cases}
\curl\curl \bE_m-\omega^2n(x) \bE_m=0\quad &~\mbox{in}~ \Omega\backslash \bar{D},\\
\curl\curl \bE_0-\omega^2 \bE_0=0\quad &~\mbox{in} ~\Omega,\\
\nu\times \bE_m =0\quad &~\mbox{on}~\partial D,\\
\nu\times \bE_m=\nu\times \bE_0,~~\nu\times\curl\bE_m=\nu\times\curl\bE_0\quad &~\mbox{on} ~\partial\Omega.
\end{cases}
\label{u1}
\end{equation} 

For any sufficiently small $\varepsilon\in\mathbb{R}^+$, by the denseness property of Maxwell-Herglotz equations, there exists $(\bE^g_{\omega},\bH^g_{\omega})$ such that
\begin{equation*}
\|\bE^g_{\omega}-\bE_0\|_{H({\rm curl}, \Omega)}<\varepsilon,~~\|\bH^g_{\omega}-\bH_0\|_{H({\rm curl}, \Omega)}<\varepsilon.
\end{equation*}

Let $(\bE^g_{\omega},\bH^g_{\omega})$ be the incident electromagnetic wave on $(\Omega\backslash\bar{D};\epsilon_m,\mu_m,0)$, and $(\bE^s_{\omega},\bH^s_{\omega})$ be the corresponding scattered wave. Then the total wave $\bE_{\omega}$ satisfy
\begin{equation*}
\begin{cases}
\curl\curl \bE_{\omega}-\omega^2n(x) \bE_{\omega}=0\quad &~\mbox{in}~ \mathbb{R}^3\backslash \bar{D},\\
\nu\times \bE_{\omega}=0\quad &~\mbox{on}~\partial D,\\
\bE_{\omega}^s \mbox{ satisfy the Silver-M$\ddot{\mbox{u}}$ller radiation condition},
\end{cases}
\label{scatter}
\end{equation*}
where $n(x)=\epsilon(x)$ in $\Omega\backslash \bar{D}$ and $n(x)=1$ in $\mathbb{R}^3\backslash \bar{\Omega}$.

Set 
\begin{equation}
\bE^s=
\begin{cases}
\bE_m-\bE_0\quad &~\mbox{in}~ \Omega\backslash \bar{D},\\
0\quad &~\mbox{in}~ \mathbb{R}^3\backslash \bar{\Omega}.
\end{cases}
\label{Es0}
\end{equation}
Then we have 
\begin{equation}
\begin{cases}
\curl\curl \bE^s-\omega^2n(x) \bE^s=\omega^2(n(x)-1)\bE_0\quad &~\mbox{in}~ \mathbb{R}^3\backslash \bar{D},\\
\nu\times \bE^s=-\nu\times \bE_0\quad &~\mbox{on}~\partial D,\\
\bE^s\mbox{ satisfy the Silver-M$\ddot{\mbox{u}}$ller radiation condition}.
\end{cases}
\label{hatE}
\end{equation}
On the other hand, 
\begin{equation}
\begin{cases}
\curl\curl \bE^s_{\omega}-\omega^2n(x)\bE^s_{\omega}=\omega^2(n(x)-1)\bE^g_{\omega}\quad &~\mbox{in}~ \mathbb{R}^3\backslash \bar{D},\\
\nu\times \bE^s_{\omega}=-\nu\times \bE^g_{\omega}\quad &~\mbox{on}~\partial D,\\
\bE^s_{\omega}\mbox{ satisfy the Silver-M$\ddot{\mbox{u}}$ller radiation condition}.
\end{cases}
\label{sE}
\end{equation}
Subtracting (\ref{sE}) from (\ref{hatE}), we have
\begin{equation}
\begin{cases}
\curl\curl (\bE^s-\bE^s_{\omega})-\omega^2n(x) (\bE^s-\bE^s_{\omega})=\omega^2(n(x)-1)(\bE_0-\bE^g_{\omega})\quad &~\mbox{in}~ \mathbb{R}^3\backslash \bar{D},\\
\nu\times (\bE^s-\bE^s_{\omega})=-\nu\times (\bE_0-\bE^g_{\omega})\quad &~\mbox{on}~\partial D,\\
\bE^s-\bE^s_{\omega}\mbox{ satisfy the Silver-M$\ddot{\mbox{u}}$ller radiation condition}.
\end{cases}
\label{Esomega}
\end{equation}

By the well-posedness of (\ref{Esomega}), (see for example, \cite{HL, LRX}), we have the following estimate
\begin{equation*}
\|\bE^s-\bE^s_{\omega}\|_{H({\rm curl}, \mathbb{R}^3\backslash \bar{D})}\leq C_1\|\bE_0-\bE^g_{\omega}\|_{H({\rm curl}, \Omega\backslash \bar{D})}+C_2\|\nu\times(\bE_0-\bE^g_{\omega})\|_{TH^{-1/2}_{{\rm Div}}(\partial D)},
\end{equation*}
where $C_1$ and $C_2$ are positive constants depending only on $\omega,\epsilon_m,\mu_m$ and $\Omega,D$. It follows from (\ref{EH}) and the Sobolev trace theorem that
$$\|\bE^s-\bE^s_{\omega}\|_{H({\rm curl}, \mathbb{R}^3\backslash \bar{D})}\leq C_3\varepsilon,$$
where $C_3$ is a positive constant depending only on $\omega,\epsilon_m,\mu_m$ and $\Omega,D$.
Consequently, one has
\beq
\|\nu\times (\bE^s-\bE_{\omega}^s)\|_{TH^{-1/2}_{{\rm Div}}(\partial\Omega)}\leq C \varepsilon,
\label{pec}
\eeq
where $C$ is a positive constant depending only on $\omega,\epsilon_m,\mu_m$ and $\Omega,D$. Therefore, by (\ref{Es0}) and the well-posedness of the scattering problem from a perfectly conducting electric obstacle, one has (\ref{nearE}). Inequality (\ref{nearH}) follows immediately from \eqref{nearE} since $\bH_{\infty}=\nu\times\bE_{\infty}$ on $\mathbb{S}^2$.
\end{proof}

Let $\mathcal{T}_{\Omega\backslash\bar{D}}\subset \mathbb{R}$ denote the set of all the interior transmission eigenvalues of (\ref{u0}) associated with $(\Omega\backslash\bar{D};\epsilon_m,\mu_m,0)$, and set
\begin{equation*}
\begin{split}
&\mathcal{W}_{\Omega\backslash\bar{D};\epsilon_m,\mu_m,0}:=\bigcup_{\omega\in\mathcal{T}_{\Omega\backslash\bar{D}}}\Big\{(\bE_0,\bH_0):(\bE_m,\bH_m),(\bE_0,\bH_0)\in H(\mbox{curl},\Omega\backslash\bar{D})\times H(\mbox{curl},\Omega) \mbox{ are a pair of} \\
&\mbox{ interior transmission eigenfunctions of (\ref{u0}) corresponding to}~\omega\mbox{  associated with }(\Omega\backslash\bar{D};\epsilon_m,\mu_m,0)\Big\}.
\end{split}
\end{equation*}
Clearly, $\mathcal{W}_{\Omega\backslash\bar{D};\epsilon_m,\mu_m,0}$ is a subspace of $W(\Omega):=\cup_{\omega\in\mathbb{R}^+}W_\omega(\Omega)$. For a sufficiently small $\varepsilon\in\mathbb{R}^+$, we let $$\mathcal{H}^\varepsilon_{\Omega\backslash\bar{D};\epsilon_m,\mu_m,0}\subset \mathcal{H}(\Omega):=\cup_{\omega\in\mathbb{R}^+}\mathcal{H}_{\omega}(\Omega)$$ be an $\varepsilon$-net of $\mathcal{W}_{\Omega\backslash\bar{D};\epsilon_m,\mu_m,0}$ in the norm $\|\cdot\|_{H(\mbox{curl},\Omega)}$. By Proposition~\ref{nearcloaking}, one clearly has that
\begin{thm}
For any $(\bE^g_{\omega},\bH^g_{\omega})\in \mathcal{H}^\varepsilon_{\Omega\backslash\bar{D};\epsilon_m,\mu_m,0}$ associated with an interior transmission eigenvalue $\omega\in\mathcal{T}_{\Omega\backslash\bar{D}}$, then there holds
\begin{equation*}
|\bE_{\infty}(\hat{x},(\bE^g_\omega,\bH^g_\omega),(\Omega\backslash\bar{D};\epsilon_m,\mu_m,0))|\leq C\varepsilon,~ |\bH_{\infty}(\hat{x},(\bE^g_\omega,\bH^g_\omega),(\Omega\backslash\bar{D};\epsilon_m,\mu_m,0))|\leq C\varepsilon,~\forall \hat{x}\in \mathbb{S}^{2},
\end{equation*}
where $C$ is a positive constant depending only on $\omega,\epsilon_m,\mu_m$ and $\Omega,D$.
\label{cloaking}
\end{thm}
Therefore, by Theorem \ref{cloaking}, the cloaked PEC obstacle $D$ together with the coating $(\Omega\backslash\bar{D};\epsilon_m,\mu_m,0)$ is nearly invisible to the wave interrogation for any incident field from $\mathcal{H}^\varepsilon_{\Omega\backslash\bar{D};\epsilon_m,\mu_m,0}$ . 
\section{Isotropic invisibility cloaking with a lossy layer}
In this section, we consider a more realistic cloaking construction other than assuming that the cloaked region $D$ is insulating. The cloaking device takes a three-layer structure with a cloaked region, a lossy layer and a cloaking shell medium (see Figure \ref{fig:three_layer}). The target medium in the cloaked region can be arbitrary but regular, whereas the media in the lossy layer and the cloaking shell are both non-singular and isotropic. 

Consider the following exterior Maxwell system
\begin{equation}
\begin{cases}
\nabla\times \bE_0(x)-i\omega\bH_0(x)=0 \quad &~\mbox{in}~ \mathbb{R}^3\backslash \bar{\Omega},\\
\nabla\times \bH_0(x)+i\omega \bE_0(x)=0 \quad &~\mbox{in}~ \mathbb{R}^3\backslash \bar{\Omega},\\
\nu\times \bE_0=f \quad &~\mbox{on} ~\partial\Omega,\\
\bE_0,\bH_0\mbox{ satisfy the Silver-M$\ddot{\mbox{u}}$ller radiation condition}.
\end{cases}
\label{exe}
\end{equation}
Define the exterior boundary impedance map as
\begin{equation}
\Lambda^{e,\omega}_{\Omega}(\nu\times \bE_0|_{\partial\Omega})=\nu\times \bH_0|_{\partial\Omega}:TH^{-1/2}_{{\rm Div}}(\partial\Omega)\rightarrow TH^{-1/2}_{{\rm Div}}(\partial\Omega),
\label{mapexterior}
\end{equation}
Since the exterior scattering problem (\ref{exe}) is well-posed, one clearly has that $\Lambda^{e,\omega}_{\Omega}$ is well-defined and moreover, it is invertible.

Let $\Sigma\Subset D$ be a domain of Lipschitz class. Let $\tau\in\mathbb{R}^+$ be an asymptotically small parameter. Set
$$\epsilon_l=\alpha_1\tau^{-1}\cdot\bI_{3\times 3},~~ \mu_l=\alpha_2\tau\cdot\bI_{3\times 3},~~ \sigma_l=\alpha_3\tau^{-1}\cdot\bI_{3\times 3},$$
where $\alpha_1,\alpha_2,\alpha_3$ are constants in $\mathbb{R}^+$. Consider an electromagnetic medium configuration mentioned in the introduction section as follows:
\beq
(\mathbb{R}^3;\epsilon,\mu,\sigma) = (\Sigma;\epsilon_a,\mu_a,\sigma_a)\wedge(D\backslash\bar{\Sigma};\epsilon_l,\mu_l,\sigma_l)\wedge(\Omega\backslash\bar{D};\epsilon_m,\mu_m,0)\wedge (\mathbb{R}^3\backslash\bar{\Omega};\bI_{3\times 3},\bI_{3\times 3},0),
\label{medium}
\eeq
where $(\Sigma;\epsilon_a,\mu_a,\sigma_a)$ is a regular electromagnetic medium (see Figure \ref{fig:three_layer}). Then, we present the proof for the main Theorem \ref{invi-cloaking}. 
\begin{proof}[Proof of Theorem~\ref{invi-cloaking}]
Let us consider the scattering system (\ref{max}) corresponding to $(\mathbb{R}^3,\epsilon,\mu,\sigma)$ described in (\ref{medium}). We prove Theorem \ref{invi-cloaking} in the following four steps.

\textbf{Step 1}: It is first noted that $(\bE,\bH)\in H(\mbox{curl},\Omega)$ satisfies the following Maxwell system
\beq
\begin{cases}
\curl \bE(x)-i\omega \mu(x)\bH(x)=0,\quad &~x\in\Omega,\\
\curl \bH(x)+i\omega \epsilon(x)\bE(x)=\sigma(x)\bE(x),\quad &~x\in\Omega,
\end{cases}
\label{ss}
\eeq
Using integration by parts, one can calculate as follows
\beq
\begin{split}
\int_{\Omega}\sigma\bE\cdot\bar{\bE}\diff x&=\int_{\Omega}\big(\curl \bH+i\omega \epsilon\bE\big)\cdot\bar{\bE}\diff x\\
&=\int_{\Omega}\bH\cdot\overline{(\curl \bE})-\int_{\partial\Omega}\bH\cdot\overline{(\nu\times\bE)}\diff s+i\omega\int_{\Omega}\epsilon\bE\cdot\bar{\bE}\diff x\\
&=-i\omega \int_{\Omega}\mu\bH\cdot\bar{\bH}\diff x-\int_{\partial\Omega}\bH\cdot\overline{(\nu\times\bE)}\diff s+i\omega\int_{\Omega}\epsilon\bE\cdot\bar{\bE}\diff x.
\end{split}
\label{sss}
\eeq
By taking the real parts of both sides of (\ref{sss}), we have
\beq
\int_{\Omega}\sigma\bE\cdot\bar{\bE}\diff x=-\mbox{Re}\left\{\int_{\partial\Omega}\bH\cdot\overline{(\nu\times\bE)}\diff s\right\}.
\label{real}
\eeq
Using
$$\bE=\bE^{i,\omega}+\bE^{s,\omega},~\bH=\bH^{i,\omega}+\bH^{s,\omega},$$
and $$\nu\times\bH^{s,\omega}|_{\partial\Omega}=\Lambda^{e,\omega}_{\Omega}(\nu\times\bE^{s,\omega}|_{\partial\Omega}),$$
one can calculate that
\begin{multline}\label{s4}
\int_{\partial\Omega}\bH\cdot\overline{(\nu\times\bE)}\diff s=\int_{\partial\Omega}\bH^{i,\omega}\cdot\overline{(\nu\times\bE^{i,\omega})}\diff s+\int_{\partial\Omega} \left( \Lambda^{e,\omega}_{\Omega}(\nu\times\bE^{s,\omega}) \times \nu \right)\cdot\overline{(\nu\times\bE^{i,\omega})}\diff s\\
+\int_{\partial\Omega}\bH^{i,\omega}\cdot\overline{(\nu\times\bE^{s,\omega})}\diff s+\int_{\partial\Omega} \left( \Lambda^{e,\omega}_{\Omega}(\nu\times\bE^{s,\omega}) \times \nu \right)\cdot\overline{(\nu\times\bE^{s,\omega})}\diff s.
\end{multline}
Using integration by parts and straightforward calculations, one has
\begin{align} \label{eq:3}
\int_{\partial\Omega}\bH^{i,\omega}\cdot\overline{(\nu\times\bE^{i,\omega})}\diff s=i\omega\int_{\Omega}|\bE^{i,\omega}|^2\diff x-i\omega\int_{\Omega}|\bH^{i,\omega}|^2\diff x.
\end{align}
Clearly, we have
\beq
\int_{\Omega}\sigma\bE\cdot\bar{\bE}\diff x\geq \alpha_3\tau^{-1}\|\bE\|^2_{L^2(D\backslash\bar{\Sigma})}.
\label{s5}
\eeq
Finally, by combining (\ref{real})--(\ref{s5}) and using the fact that the skew-symmetric bilinear form
\begin{align*}
  \mathcal{B}(\mathbf{j}, \mathbf{m}) = \int_{\partial\Omega} \mathbf{j} \cdot (\mathbf{m} \times \nu) \, {\rm d}s
\end{align*}
defines a non-degenerate duality product on $TH^{-1/2}_{{\rm Div}}(\partial\Omega) \times TH^{-1/2}_{{\rm Div}}(\partial\Omega)$ \cite{CL}, we have
\begin{multline}\label{E1}
\|\bE\|^2_{L^2(D\backslash\bar{\Sigma})} \leq C\tau \Big( \|\bE^g_{\omega}\|_{H(\mbox{curl},\Omega)}\|\Lambda^{e,\omega}_{\Omega}(\nu\times\bE^{s,\omega})\|_{TH^{-1/2}_{{\rm Div}}(\partial\Omega)} \\
+\|\bH^g_{\omega}\|_{H(\mbox{curl},\Omega)}\|\nu\times\bE^{s,\omega}\|_{TH^{-1/2}_{{\rm Div}}(\partial\Omega)}
+\|\nu\times\bE^{s,\omega}\|_{TH^{-1/2}_{{\rm Div}}(\partial\Omega)}\|\Lambda^{e,\omega}_{\Omega}(\nu\times\bE^{s,\omega})\|_{TH^{-1/2}_{{\rm Div}}(\partial\Omega)} \Big),
\end{multline}

where $C$ is a positive constant depending only on $\Omega$ and $\alpha_3,\omega$.

\textbf{Step 2}: We want to show that
\beq
\|\nu\times\bE\|_{TH^{-1/2}_{{\rm Div}}(\partial D)}\leq C\|\bE\|_{L^2(D\backslash\bar{\Sigma})},
\label{E2}
\eeq
where $C$ depends only on $\Omega$ and $\alpha_3,\omega$. We shall make use of the following duality relation\cite{BL},
\beq
\|\nu\times\bE\|_{TH^{-1/2}_{{\rm Div}}(\partial D)}=\sup_{\|\bm{\psi}\|_{TH^{-1/2}_{{\rm Curl}}(\partial\Omega)}\leq 1}\Big|\int_{\partial\Omega}(\nu\times\bE)\cdot \bm{\psi}\diff s\Big|,
\label{E3}
\eeq
where
$$TH^{-1/2}_{{\rm Curl}}(\partial D):=\big\{ U\in TH^{-1/2}(\partial D):{\rm Curl}\, U\in H^{-1/2}(\partial D)\big\}.$$
For any $\bm{\psi} \in TH^{-1/2}_{{\rm Curl}}(\partial D)$, there exists $\bF\in H^2(\mbox{curl},D)$ such that (see Lemma 3.5 in \cite{BLZ})
\begin{enumerate}
\item $\nu\times \bF=0$ on $\partial D$;
\item $\nu\times\nu\times\curl \bF=\nu\times\nu\times \bm{\psi}$ on $\partial D$;
\item $\|\bF\|_{H^2(\mbox{curl},D)}\leq C\|\bm{\psi}\|_{TH^{-1/2}_{{\rm Curl}}(\partial D)}$, where $C$ depends only on $D$;
\item $\bF=0$ in $\Sigma$.
\end{enumerate}
By virtue of the duality relation (\ref{E3}) and using the auxiliary function $\bF$, along with the integration by parts, we have
\beq
\begin{split}
\int_{\partial D}(\nu\times\bE)\cdot \psi\diff s&=\int_{\partial D}(\nu\times\bE)\cdot \curl \bF \diff s\\
&=\int_{\partial D} (\nu\times\bE)\cdot \curl \bF \diff s - \int_{\partial D} (\nu\times\bF)\cdot \curl \bE \diff s\\
&=\int_{D}(\curl\curl \bE) \cdot \bF \diff x-\int_{D}(\curl\curl \bF)\cdot \bE \diff x.
\end{split}
\label{EF}
\eeq
Noting that in $D\backslash\bar{\Sigma}$, one has
\begin{equation*}
\begin{cases}
\curl \bE(x)-i\omega \alpha_2\tau\bH(x)=0,\quad &~x\in D\backslash\bar{\Sigma},\\
\curl \bH(x)+i\omega \alpha_1\tau^{-1}\bE(x)=\alpha_3\tau^{-1}\bE(x),\quad &~x\in D\backslash\bar{\Sigma},
\end{cases}
\end{equation*}
one has by direct verifications that
\beq
\curl\curl \bE=i\omega\alpha_2(\alpha_3-i\omega\alpha_1)\bE~~\mbox{in}~D\backslash\bar{\Sigma}.
\label{cc}
\eeq
Plug (\ref{cc}) into (\ref{EF}), one has
\beq
\begin{split}
\Big|\int_{\partial D}(\nu\times\bE)\cdot \bm{\psi} \diff s\Big|&=\left|i\omega\alpha_2(\alpha_3-i\omega\alpha_1)\int_{D}\bE\cdot \bF\diff x - \int_{D}\bE\cdot(\curl\curl \bF)\diff x \right|\\
&\leq C\|\bE\|_{L^2(D)}\|\bF\|_{H^2(\mbox{curl},D)}\\
&\leq C\|\bE\|_{L^2(D)}\|\bm{\psi}\|_{TH^{-1/2}_{{\rm Curl}}(\partial D)},
\end{split}
\label{EP}
\eeq
where $C$ depends only on $\alpha_1,\alpha_2,\alpha_3$ and $D$. Finally, by combining (\ref{E3}) and (\ref{EP}), one immediately has (\ref{E2}). 

\textbf{Step 3}: By (\ref{E1}) and (\ref{E2}) and the boundedness of the map $\Lambda_\Omega^{e,w}$, we have
\begin{multline}\label{EE}
\|\nu\times\bE\|_{TH^{-1/2}_{{\rm Div}}(\partial D)}
\leq C\tau^{1/2} \Big( \|\bE^g_{\omega}\|^{1/2}_{H(\mbox{curl},\Omega)}\|\Lambda^{e,\omega}_{\Omega}(\nu\times\bE^{s,\omega})\|^{1/2}_{TH^{-1/2}_{{\rm Div}}(\partial\Omega)}\\
+\|\bH^g_{\omega}\|^{1/2}_{H(\mbox{curl},\Omega)}\|\nu\times\bE^{s,\omega}\|^{1/2}_{TH^{-1/2}_{{\rm Div}}(\partial\Omega)}
+\|\nu\times\bE^{s,\omega}\|^{1/2}_{TH^{-1/2}_{{\rm Div}}(\partial\Omega)}\|\Lambda^{e,\omega}_{\Omega}(\nu\times\bE^{s,\omega})\|^{1/2}_{TH^{-1/2}_{{\rm Div}}(\partial\Omega)} \Big) \\
\leq C_1\tau^{1/2} \Big( \|\bE^g_{\omega}\|_{H(\mbox{curl},\Omega)}+\|\bH^g_{\omega}\|_{H(\mbox{curl},\Omega)}+\|\nu\times\bE^{s,\omega}\|_{TH^{-1/2}_{{\rm Div}}(\partial\Omega)} \Big),
\end{multline}
where $C_1$ is another constant depending only on $\Omega,D$ and $\omega,\epsilon,\mu$ and $\sigma$.

Let $\omega\in\mathbb{R}^+$ be an interior transmission eigenvalue associated with $(\Omega\backslash\bar{D};\epsilon_m,\mu_m,0)$. It is easily seen that
\begin{equation*}
\begin{cases}
\curl \bE-i\omega\mu\bH=0 \quad &~\mbox{in}~ \mathbb{R}^3\backslash \bar{D},\\
\curl \bH+i\omega\epsilon \bE=0 \quad &~\mbox{in}~ \mathbb{R}^3\backslash \bar{D},\\
\nu\times \bE= \nu \times \bE \quad &~\mbox{on}~ \partial D,\\
\bE=\bE^{s,\omega}+\bE^g_{\omega},~\bH=\bH^{s,\omega}+\bH^g_{\omega} \quad &~\mbox{in}~ \mathbb{R}^3\backslash \bar{D},\\
\bE^{s,\omega},\bH^{s,\omega} \mbox{ satisfy the Silver-M$\ddot{\mbox{u}}$ller radiation condition},
\end{cases}
\end{equation*}
where $\mu:=\mu_m\chi(\Omega\backslash \bar{D})+1\chi(\mathbb{R}^3\backslash \bar{\Omega})$, $\epsilon:=\epsilon_m\chi(\Omega\backslash \bar{D})+1\chi(\mathbb{R}^3\backslash \bar{\Omega})$.

We also introduce $\bE^{s,\omega}_c=\bE_c-\bE^g_{\omega}$ and $\bH^{s,\omega}_c=\bH_c-\bH^g_{\omega}$ satisfying
\begin{equation*}
\begin{cases}
\curl \bE_c-i\omega\mu\bH_c=0 \quad &~\mbox{in}~ \mathbb{R}^3\backslash \bar{D},\\
\curl \bH_c+i\omega\epsilon \bE_c=0 \quad &~\mbox{in}~ \mathbb{R}^3\backslash \bar{D},\\
\nu\times \bE_c=0 \quad &~\mbox{on}~ \partial D,\\
\bE_c=\bE^{s,\omega}_c+\bE^g_{\omega},~\bH_c=\bH_c^{s,\omega}+\bH^g_{\omega} \quad &~\mbox{in}~ \mathbb{R}^3\backslash \bar{D},\\
\bE^{s,\omega}_c,\bE^{s,\omega}_c\mbox{ satisfy the Silver-M$\ddot{\mbox{u}}$ller radiation condition}.
\end{cases}
\end{equation*}
Set
$$\hat{\bE}=\bE-\bE_c~\mbox{and }\hat{\bH}=\bH-\bH_c~~\mbox{in }  \mathbb{R}^3\backslash \bar{D}.$$
Then $(\hat{\bE},\hat{\bH})$ satisfies
\begin{equation*}
\begin{cases}
\curl \hat{\bE}-i\omega\mu\hat{\bH}=0\quad &~\mbox{in}~ \mathbb{R}^3\backslash \bar{D},\\
\curl \hat{\bH}+i\omega\epsilon \hat{\bE}=0 \quad &~\mbox{in}~ \mathbb{R}^3\backslash \bar{D},\\
\nu\times \hat{\bE}=\nu \times \bE \quad &~\mbox{on}~ \partial D,\\
\hat{\bE},\hat{\bH}\mbox{ satisfy the Silver-M$\ddot{\mbox{u}}$ller radiation condition}.
\end{cases}
\end{equation*}
By Lemma \ref{lem} in the following, we have
$$\|\nu\times \hat{\bE}\|_{TH^{-1/2}_{{\rm Div}}(\partial \Omega)}\leq C\|\nu \times \bE\|_{TH^{-1/2}_{{\rm Div}}(\partial D)},$$
where $C$ is a positive constant depending only on $D$, $\Omega$ and $\omega$, $\mu,\epsilon,\sigma$.

Note that  $$\hat{\bE}=\bE-\bE_c=\bE^{s,\omega}-\bE_c^{s,\omega}~~\mbox{and }\hat{\bH}=\bH-\bH_c=\bH^{s,\omega}-\bH^{s,\omega}_c~~\mbox{in }  \mathbb{R}^3\backslash \bar{D}.$$ Then there is
\begin{align*}
\|\nu\times \bE^{s,\omega}-\nu\times\bE^{s,\omega}_c\|_{TH^{-1/2}_{{\rm Div}}(\partial \Omega)}
=\|\nu\times \hat{\bE}\|_{TH^{-1/2}_{{\rm Div}}(\partial \Omega)}\leq C\|\nu \times \bE\|_{TH^{-1/2}_{{\rm Div}}(\partial D)},
\end{align*}
which in turn implies that
\beq
\|\nu\times \bE^{s,\omega}\|_{TH^{-1/2}_{{\rm Div}}(\partial \Omega)}\leq \|\nu\times\bE^{s,\omega}_c\|_{TH^{-1/2}_{{\rm Div}}(\partial \Omega)}+C\|\nu \times \bE\|_{TH^{-1/2}_{{\rm Div}}(\partial D)}.
\label{sw0}
\eeq
By the argument in the proof of Theorem \ref{nearcloaking}, we see that
\beq
\|\nu\times\bE^{s,\omega}_c\|_{TH^{-1/2}_{{\rm Div}}(\partial \Omega)}\leq C\varepsilon.
\label{sw0c}
\eeq
By applying (\ref{EE}), (\ref{sw0c}) to (\ref{sw0}), we have
\beq
\begin{split}
&\|\nu\times \bE^{s,\omega}\|_{TH^{-1/2}_{{\rm Div}}(\partial \Omega)} \\
&\leq C \left( \epsilon+\tau^{1/2}\|\bE^g_{\omega}\|_{H(\mbox{curl},\Omega)}+\tau^{1/2}\|\bH^g_{\omega}\|_{H(\mbox{curl},\Omega)}+\tau^{1/2}\|\nu\times\bE^{s,\omega}\|_{TH^{-1/2}_{{\rm Div}}(\partial\Omega)} \right)
\end{split}
\label{Es}
\eeq
for some constant independent of $\varepsilon$ and $\tau$. For sufficiently small $\tau$, we obviously have from (\ref{Es}) that
\beq
\|\nu\times \bE^{s,\omega}\|_{TH^{-1/2}_{{\rm Div}}(\partial\Omega)}\leq C\Big(\varepsilon+\tau^{1/2}\|\bE^g_{\omega}\|_{H(\mbox{curl},\Omega)}+\tau^{1/2}\|\bH^g_{\omega}\|_{H(\mbox{curl},\Omega)}\Big),
\label{Es1}
\eeq
where $C$ is a constant independent of $\varepsilon,\tau, (\bE^g_{\omega},\bH^g_{\omega})$ and $\epsilon_a,\mu_a,\sigma_a$.

Finally, by the well-posedness of the electromagnetic scattering problem, one readily has (\ref{Einvi}) and (\ref{Hinvi}) from (\ref{Es1}). 

\end{proof}

The following lemma is crucial in the proof of the above theorem.
\begin{lem}
Let $(\mathbb{R}^3\backslash\bar{D};\epsilon_m,\mu_m,0)$ be the one described in (\ref{medium}). Let $(\hat{\bE},\hat{\bH})\in H(\rm{curl},\mathbb{R}^3\backslash\bar{D})$ be the unique solution to
\begin{equation*}
\begin{cases}
\curl \hat{\bE}-i\omega\mu\hat{\bH}=0 \quad &~{\rm in}~ \mathbb{R}^3\backslash \bar{D},\\
\curl \hat{\bH}+i\omega\epsilon \hat{\bE}=0 \quad &~{\rm in}~ \mathbb{R}^3\backslash \bar{D},\\
\nu\times \hat{\bE}=f \quad &~{\rm in}~ \partial D,\\
\hat{\bE},\hat{\bH}\mbox{ satisfy the Silver-M$\ddot{\mbox{u}}$ller radiation condition}.
\end{cases}
\end{equation*}
Then there holds
\beq
\|\nu\times \hat{\bE}\|_{TH^{-1/2}_{{\rm Div}}(\partial \Omega)}\leq C\|f\|_{TH^{-1/2}_{{\rm Div}}(\partial D)},
\label{otod}
\eeq
where $C$ is positive constant independent of $f$.
\label{lem}
\end{lem}

\begin{proof}
By simple calculations, we get the following identity in terms of $\hat\bE$:
$$\curl\curl\hat{\bE}-\omega^2n(x)\hat{\bE}=0~\mbox{in}~ \mathbb{R}^3\backslash \bar{D},$$
where $n(x):=\epsilon(x)\mu(x)>0$. Then, we have
$$\curl\curl\hat{\bE}-\omega^2\hat{\bE}=\omega^2(n-1)\hat{\bE}~\mbox{in}~ \mathbb{R}^3\backslash \bar{D}.$$

We seek the solution in the form of the electromagnetic field of an electric dipole distribution. 

Suppose there exists density function $\varphi(x)\in TH^{-1/2}_{{\rm Div}}(\partial D)$ such that $\hat{\bE}$ can be represented as
\beq
\hat{\bE}(x)=\curl_x\int_{\partial D}\varphi(y)\Phi(x,y)\diff s(y)+\int_{\Omega\backslash \bar{D}}\omega^2(n-1)\hat{\bE}(y)\Phi(x,y)\diff x,~x\in\Omega\backslash \bar{D},
\label{rep}
\eeq
where $\Phi$ is the fundamental solution for Helmholtz equation.
We introduce the single layer operator
$$
\mathcal{S}_D\varphi(x):=\int_{\partial D}\Phi(x,y)\varphi(y)\diff s(y),~x\in\mathbb{R}^3,
$$
and the boundary operator
$$
\mathcal{M}_D\varphi(x):=\nu(x)\times\curl_x\int_{\partial D}\Phi(x,y)\varphi(y)\diff s(y),~x\in\partial D.
$$
Then we have the following jump relation
$$\nu(x)\times\curl_x\mathcal{S}_D|_{\pm}\varphi(x)=\mathcal{M}_D\varphi(x)\pm\frac{1}{2}\varphi(x),~x\in\partial D.$$

Define
$$\mathcal{V}_N\varphi(x):=\int_{\mathbb{R}^3\backslash \bar{D}}\omega^2(n-1)\varphi(y)\Phi(x,y)\diff y,~x\in \mathbb{R}^3\backslash \bar{D}.$$
Then we have the following system with unknowns $(\hat{\bE},\varphi)\in  H(\mbox{curl},\Omega\backslash\bar{D})\times TH^{-1/2}_{{\rm Div}}(\partial D)$:
\beq
\begin{cases}
\hat{\bE}(x)&=\curl_x\mathcal{S}_D\varphi(x)+\mathcal{V}_N\hat{\bE}(x),~x\in\Omega\backslash \bar{D},\\
f(x)&=\mathcal{M}_D\varphi(x)+\frac{1}{2}\varphi(x)+\nu(x)\times\mathcal{V}_N\hat{\bE}(x),~x\in\partial D.
\end{cases}
\label{sys}
\eeq
We write system (\ref{sys}) in the matrix-vector notation as
\begin{equation*}
\begin{bmatrix}
\bI&0\\
2\nu\times\mathcal{V}_N&\bI
\end{bmatrix}
\begin{bmatrix}
\hat{\bE}\\
\varphi
\end{bmatrix}+\begin{bmatrix}
-\mathcal{V}_N & -\curl_x\mathcal{S}_D\\
0&2\mathcal{M}_D
\end{bmatrix}\begin{bmatrix}
\hat{\bE}\\
\varphi
\end{bmatrix}=\begin{bmatrix}
0\\
2f
\end{bmatrix},
\end{equation*}
where $\bI$ is the identity operator. The matrix operator
\begin{equation*}
\begin{bmatrix}
\bI&0\\
2\nu\times\mathcal{V}_N&\bI
\end{bmatrix}
\end{equation*}
is invertible, and has inverse
\begin{equation*}
\begin{bmatrix}
\bI&0\\
-2\nu\times\mathcal{V}_N&\bI
\end{bmatrix}.
\end{equation*}
All the entries of
\begin{equation*}
\begin{bmatrix}
-\mathcal{V}_N & -\curl_x\mathcal{S}_D\\
0&2\mathcal{M}_D
\end{bmatrix}
\end{equation*}
are compact(see section 6.3, \cite{CK}). Hence, we can apply the Riesz-Fredholm theory to (\ref{sys}). For this purpose, suppose $\hat{\bE}$ and $\varphi$ are a solution to (\ref{sys}) with $f=0$ on $\partial D$. By Theorem 6.11 in \cite{CK}, it is easy to see that $\hat{\bE}=0$ in $\mathbb{R}^3\backslash\bar{D}$. Then we have
\begin{equation*}
0=\mathcal{M}_D\varphi(x)+\frac{1}{2}\varphi(x),~x\in\partial D.
\end{equation*}
Hence, by Riesz-Fredholm theory, it is sufficient to show that
$$ \left( \frac{1}{2}\bI+\mathcal{M}_D \right) \varphi(x)=0,~\varphi\in TH^{-1/2}_{{\rm Div}}(\partial D)$$
has only trivial solution, i.e., $\varphi=0$. Since the null space of the operator $\frac{1}{2}\bI+\mathcal{M}_D$ corresponds to solutions of the homogeneous interior Maxwell problem(see Theorem 4.23, \cite{CK2}). If $\omega$ is not an interior Maxwell eigenvalue, then $\frac{1}{2}\bI+\mathcal{M}_D$ is invertible. Therefore, we have $\varphi=0$.

Finally, by the integral representation (\ref{rep}), direct calculations show (\ref{otod}). The proof is done.

\end{proof}

\section{Conclusion}

In this paper, we consider a novel interior transmission eigenvalue problem associated with the Maxwell system, where the inhomogeneous EM medium contains a PEC obstacle. This is mainly motivated by the study on invisibility cloaking from the EM wave probing. In certain practical scenarios, we establish the existence and discreteness of the transmission eigenvalues and eigenfunctions of the interior transmission eigenvalue problem. We would like to emphasize that our study in this aspect is not exclusive, and the existence and discreteness of the interior transmission eigenvalues and eigenfunctions may hold in other scenarios, which is definitely worth of further investigation. Using the Maxwell-Herglotz approximation to the interior transmission eigenfunctions derived previously, we can generate a set of nearly non-scattering waves corresponding to a PEC obstacle coated with a layer of regular isotropic EM medium. Finally, by introducing a deliberately designed lossy layer, we can construct a novel cloaking device that takes a three-layered structure. The innermost core is the cloaked region, where the cloaked object can be arbitrary; the outermost layer is the cloaking region, and the lossy layer lies in between the cloaking and cloaked regions. The mediums inside the cloaking region and the lossy layer are both regular and isotropic. The nearly cloaking effect is only achieved for waves from the nearly non-scattering set generated above, and we sharply quantify the cloaking performances.  

The work of Jingzhi Li was supported by the NSF of China (No.~11571161) and the Shenzhen Sci-Tech (No.~JCYJ20160530184212170). The work of Hongyu Liu was supported by the FRG fund from Hong Kong Baptist University, the Hong Kong RGC grant (No.~12302415), and the NSF of China (No.~11371115). The work of Yuliang Wang was supported by the Hong Kong RGC grant (No.~12328516) and the NSF of China (No.~11601459). 


    \end{document}